\documentclass[letterpaper, 12pt]{iopart}
\usepackage{iopams,amssymb,epsfig,subfigure,graphicx,amsthm}

\usepackage{amscd}

\def\R{\mathbb R}
\def\w{\omega}

\newtheorem{theorem}{Theorem}
\newtheorem{lemma}[theorem]{Lemma}
\newtheorem{corollary}[theorem]{Corollary}

\newtheorem{fact}{Fact}

\newtheorem{rremark}{Remark}
\newenvironment{remark}{\begin{rremark} \rm}{\end{rremark}}

\begin{document}

\title[Localized standing waves in inhomogeneous Schr\"odinger equations]{Localized standing waves in inhomogeneous Schr\"odinger equations}

\author{ R.\ Marangell \& C.K.R.T.\ Jones}
\address{Department of Mathematics, University of North Carolina at Chapel Hill, \\Chapel Hill, NC 27599\\
Warwick Mathematics Institute, University of Warwick, UK}

\author{H.\ Susanto}
\address{School of Mathematical Sciences, University of Nottingham,
University Park, Nottingham, NG7 2RD, UK}

\begin{abstract}
A nonlinear Schr\"odinger equation arising from light propagation down an inhomogeneous medium is considered. The inhomogeneity is reflected through a non-uniform coefficient of the non-linear term in the equation. In particular, a combination of self-focusing and self-defocusing nonlinearity, with the self-defocusing region localized in a finite interval, is investigated. Using numerical computations, the extension of linear eigenmodes of the corresponding linearized system into nonlinear states is established, particularly nonlinear continuations of the fundamental state and the first excited state. The (in)stability of the states is also numerically calculated, from which it is obtained that symmetric nonlinear solutions become unstable beyond a critical threshold norm. Instability of the symmetric states is then investigated analytically through the application of a topological argument. Determination of instability of positive symmetric states is reduced to simple geometric properties of the composite phase plane orbit of the standing wave. Further the topological argument is applied to higher excited states and instability is again reduced to straightforward geometric calculations.  For a relatively high norm, it is observed that asymmetric states bifurcate from the symmetric ones. The stability and instability of asymmetric states is also considered.

\end{abstract}

\maketitle

\section{Introduction}

Inhomogeneities can act as an effective trapping to collective excitations in nonlinear media. In the field of nonlinear integrated optics, the first theoretical works on guided waves by an interface between a linear and nonlinear medium appeared in \cite{lomt81,agra80,toml80,mara81}. Various stationary wave profiles propagating along nonlinear planar optical guides in a layered structure are then extensively considered (see, e.g., \cite{akhm82,jens82,seat85,boar91,{steg92}} and references therein). When the inhomogeneities are periodic, one will obtain discrete waveguide arrays, which have become an independent topic of interest \cite{lede08}. A next fundamental question is whether or not the standing waves are stable to propagation along the nonlinear waveguide.

The stability of stationary nonlinear Schrodinger waves in homogeneous media was first considered by Vakhitov and Kolokolov \cite{vakh73,kolo73}. Using variational arguments, a criterion was derived relating the soliton linear stability and the slope of the corresponding power-dispersion curve, i.e.\ the known Vakhitov-Kolokolov condition. The method was later rigorously justified by Weinstein in \cite{wein85}, and in \cite{gril87}. Subsequent studies extend the condition for various situations, including inhomogeneous problems \cite{ckrtjmoloney86,tran92a,tran92b,tran92c,tran92d,mitc93} (see also a recent brief review \cite{siva08} and references therein).

An interesting waveguide system was proposed in \cite{tran92b}, consisting of a self-focusing Schr\"odinger equation and a self-defocusing type inhomogeneity with finite length. It is experimentally feasible to fabricate such a waveguide using the current technology as self-focusing and self-defocusing can be achieved in the same medium, structure, and wavelength \cite{mora01}. In the context of a Bose-Einstein condensation \cite{bose24,eins25,ande95,davi95,brad97}, which is also modeled by a nonlinear Schr\"odinger equation \cite{gros61,pita61}, such a sign-changing nonlinearity coefficient can be created by spatially varying the condensate's atomic scattering length making the so-called collisionally inhomogeneous nonlinearity \cite{theo05,theo06,saka05}.

Tran \cite{tran92b} extended the work of, e.g., \cite{akhm82,{ckrtjmoloney86}}, in which the inhomogeneity is linear. In \cite{tran92b} the self-defocusing inhomogeneity has a small nonlinearity coefficient, such that the characteristics of the stationary solutions are closely related to the corresponding linear problem. The system is later studied by Leon \cite{jleon04}, where the nonlinearity coefficient of the inhomogeneity is of the same order as the self-focusing regions. Analytical solutions of stable symmetric solutions below a threshold amplitude are derived in terms of Jacobian elliptic functions \cite{jleon04}. Here, we revisit the problem.

We study the existence and stability of symmetric and asymmetric solutions, particularly the fundamental and the first excited mode, when the nonlinearity coefficient of the defocusing inhomogeneity is of the same order as the focusing bounding regions. We show that continuing from the linear limit solutions, there is a saddle-node bifurcation at which the symmetric mode becomes unstable and asymmetric modes emerge. Even though it is similar to the results reported in \cite{tran92b}, there is a significant difference where the asymmetric positive solutions are all stable in their existence region. Moreover, \cite{tran92b} only considers positive solutions. Besides determining the instability of symmetric solutions numerically, we also show it analytically using topological argument techniques as developed in \cite{ckrtjmoloney86,{ckrtj88}}. We also comment on the inapplicability of the analytical techniques to asymmetric solutions.

In Section 2, the governing equations are discussed and the corresponding linear eigenvalue problem is derived. In Section 3, we consider the linear limit of the equations, where a transcendental equation determining the bifurcation points of nonuniform solutions from the uniform solution $u=0$ is derived. In the same section, numerical continuations of the fundamental and the first excited state from the linear limit to nonlinear states are presented. The linear (in)stability of the numerically obtained (symmetric and asymmetric) solutions is then determined numerically by solving the corresponding linear eigenvalue problem. The instability of the symmetric solutions are analyzed analytically in Section 4 using a topological argument. In Section 5 we consider some asymmetric solutions.

\section{Mathematical model}

We consider the following governing system of differential equations
\begin{equation}\begin{array}{lll}
i\Psi_{t}+\Psi_{xx}+|\Psi|^2\Psi=V\Psi &&|x|>L,\\
i\Psi_{t}+\Psi_{xx}-\eta |\Psi|^2\Psi=0 &&|x|<L,
\end{array}
\label{gov1}
\end{equation}
where the `outer' and the `inner' equation has focusing and defocusing ($\eta>0$) type nonlinearity, respectively, and $L$ is a positive real parameter representing half the length of the waveguide. The norm
\[
N=\int_{-\infty}^{\infty}|\Psi(x,t)|^2\,dx,
\]
which is physically related to the intensity power of the electromagnetic field in the context of nonlinear optics or the number of atoms in Bose-Einstein condensates is conserved.

To study standing waves of (\ref{gov1}), we pass to a rotating
frame and consider solutions of the form $\Psi(x,t) = e^{-i \w t} \psi(x,t)$. We then have
\begin{equation}\begin{array}{lll}
i\psi_{t}+\psi_{xx}+|\psi|^2\psi=(V-\w) \psi&&|x|>L,\\
i\psi_{t}+\psi_{xx}-\eta |\psi|^2\psi=-\w \psi &&|x|<L.
\end{array}
\label{gov2}
\end{equation}
Standing wave solutions of (\ref{gov1}) will be steady-state solutions to (\ref{gov2}). In the following, the parameter $\eta$ is taken to be $\eta = 1$. We consider real, $t$ independent solutions $u(x)$ to the ODE:
\begin{equation}\begin{array}{ccccc}
u_{xx} &=& (V-\w )u - u^3 & & |x|>L, \\
u_{xx} &=& -\w u + u^3 & & |x| < L.
\end{array} \label{stat1}\end{equation}
\noindent To obtain solutions that decay to 0 as $x \to \pm \infty$, the condition that $V-\w > 0$ is required, with $\w \in \R$. We will also require that $u_x \to 0$ as $x \to \pm \infty$. To establish the instability of a standing wave solution we linearize (\ref{gov2}) about a solution to (\ref{stat1}). Writing $\psi=u(x) + \epsilon\left((r(x)+is(x))e^{\lambda t} + (r(x)^\star+is(x)^\star) e^{\lambda^\star t}\right)$ and retaining terms linear in $\epsilon$ leads to the eigenvalue problem
\begin{equation}
\lambda\left(\begin{array}{cc} r \\ s \end{array}\right) = \left(\begin{array}{cc} 0 & D_{-} \\ - D_{+} & 0 \end{array}\right) \left(\begin{array}{cc} r \\ s \end{array}\right)  = M \left(\begin{array}{cc} r \\ s \end{array}\right),
\label{eq:linear}
\end{equation}
where the linear operators $D_{+}$ and $D_{-}$ are defined as
\begin{eqnarray}
\begin{array}{lll}
D_{+} = \begin{array}{lll}
   \frac{\partial^2}{\partial x^2} - (V- \w ) + 3u^2, & |x|>L, \\
   \frac{\partial^2}{\partial x^2} + \w - 3u^2, &|x|<L,
  \end{array}
\end{array}
\label{eq:plusoperator}\\
\begin{array}{lll}
D_{-} = \begin{array}{lll}
   \frac{\partial^2}{\partial x^2} - (V- \w ) + u^2, & |x|>L, \\
   \frac{\partial^2}{\partial x^2} + \w - u^2, &|x|<L.
  \end{array}
\end{array}
\label{eq:minusoperator}
\end{eqnarray}
It is then clear that the presence of an eigenvalue of $M$ with positive real part implies instability.

\section{Linear states and their continuation}

\begin{figure}[tbhp!]
\begin{center}
\includegraphics[scale=0.5]{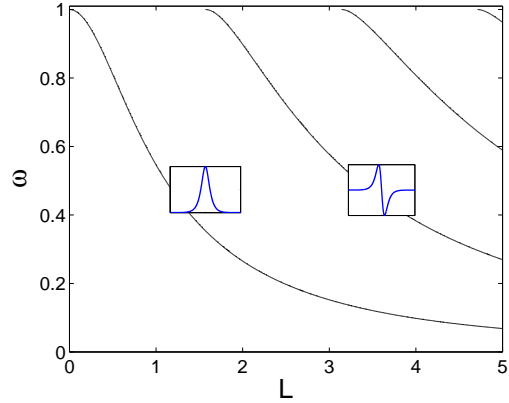}
\caption{Bifurcation points of non-uniform solutions from the zero solution $u\equiv0$ in the $(L,\omega)$-plane for $V=1$. The insets present a sketch of the corresponding solution $u(x)$ along the first two branches.}
\label{transfig}
\end{center}
\end{figure}

In the small limit of $u(x)$, the governing equation (\ref{stat1}) is reduced to the linearized system
\begin{equation}\begin{array}{ccccc}
u_{xx} &=& (V-\w )u  & & |x|>L, \\
u_{xx} &=& -\w u  & & |x| < L,
\end{array} \label{lin}\end{equation}
which can be simply solved analytically to yield
\begin{equation}
u(x)=\left\{\begin{array}{lllllll}
& e^{-\sqrt{V-\w}|x|},  &  x<-L, \\
& c_e\cos(\sqrt\omega x)+c_o\sin(\sqrt{\omega}x), &  |x| < L,\\
& c_r e^{-\sqrt{V-\w}|x|},  & x>L.
\end{array} \right.
\label{linsol}
\end{equation}
From the natural continuity conditions at the points of discontinuity
\[u(\pm L^+)=u(\pm L^-),\,u_x(\pm L^+)=u_x(\pm L^-),\]
one will obtain that the parameters of the linear states above will have to satisfy the transcendental equation
\begin{equation}
\sqrt{V-\omega}\left(1-2\cos^2(\sqrt{\omega}L)\right)=\frac12\left(\frac{V}{\sqrt\omega}-2\sqrt\omega\right)\sin(2\sqrt\omega L).
\label{trans}
\end{equation}
This equation determines bifurcation points of non-uniform states from the zero solution. A plot of (\ref{trans}) for $V=1$ is given in Fig.\ \ref{transfig}.

Starting from a bifurcation point, as the parameter $\omega$ varies, the corresponding linear limit solution will deform and nonlinear terms will play a role. Even though one can still represent the continued solutions in terms of the Jacobian elliptic functions \cite{jleon04}, here we solely use numerical computations. A pseudo-arclength method is used to follow the existence curve of a solution as a parameter is varied. We have solved Eqs.\ (\ref{stat1})--(\ref{eq:linear}) numerically to study the existence and the stability of localized standing waves, where a central finite difference is used to approximate the Laplacian with a relatively fine discretization. In particular, we consider the first two branches of linear limits shown in Fig.\ \ref{transfig}.

\subsection{Positive solutions}

\begin{figure}[tbhp!]
\begin{center}
\includegraphics[scale=0.5]{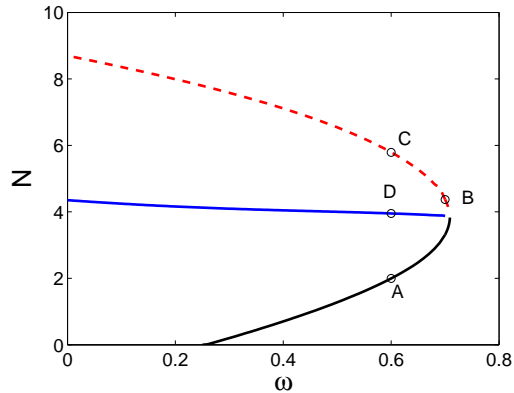}
\caption{(Color online) The norm $P$ as a function of $\omega$ for the first (positive) state corresponding to the first branch in Fig.\ \ref{transfig}. As $\omega$ increases from the bifurcation point $\omega\approx0.265$, there is a bifurcation at which the state becomes unstable. In addition to the symmetric positive states, there is also a stable asymmetric state along the middle (solid blue) branch. Solid and dashed curve represents stable and unstable solutions, respectively.}
\label{fig2_num}
\end{center}
\end{figure}

First, we consider the continuation of the linear state corresponding to the first branch. For illustrative purposes, we take $V=1$ and $L=2$, i.e.\ the fundamental state mode originates from $\omega\approx0.265$. In Fig.\ \ref{fig2_num}, we present the numerically obtained continuation of the linear positive solution as $\omega$ varies.

At the bifurcation point, the linear state is expected to be stable, similar to the zero uniform state $u(x)\equiv0$. As $\omega$ increases, the norm $N$ of the solution increases as well. By appealing to the work of  \cite{gril87,{sh-str85},{wein85}}, we obtain that the solution along this branch is stable. In Fig.\ \ref{fig1_num}(a), we depict a solution corresponding to point A in Fig.\ \ref{fig2_num} and its eigenvalue structure in the complex plane, where one can see that all the eigenvalues are on the imaginary line.

\begin{figure}[tbhp!]
\begin{center}
\subfigure[]{\includegraphics[scale=0.45]{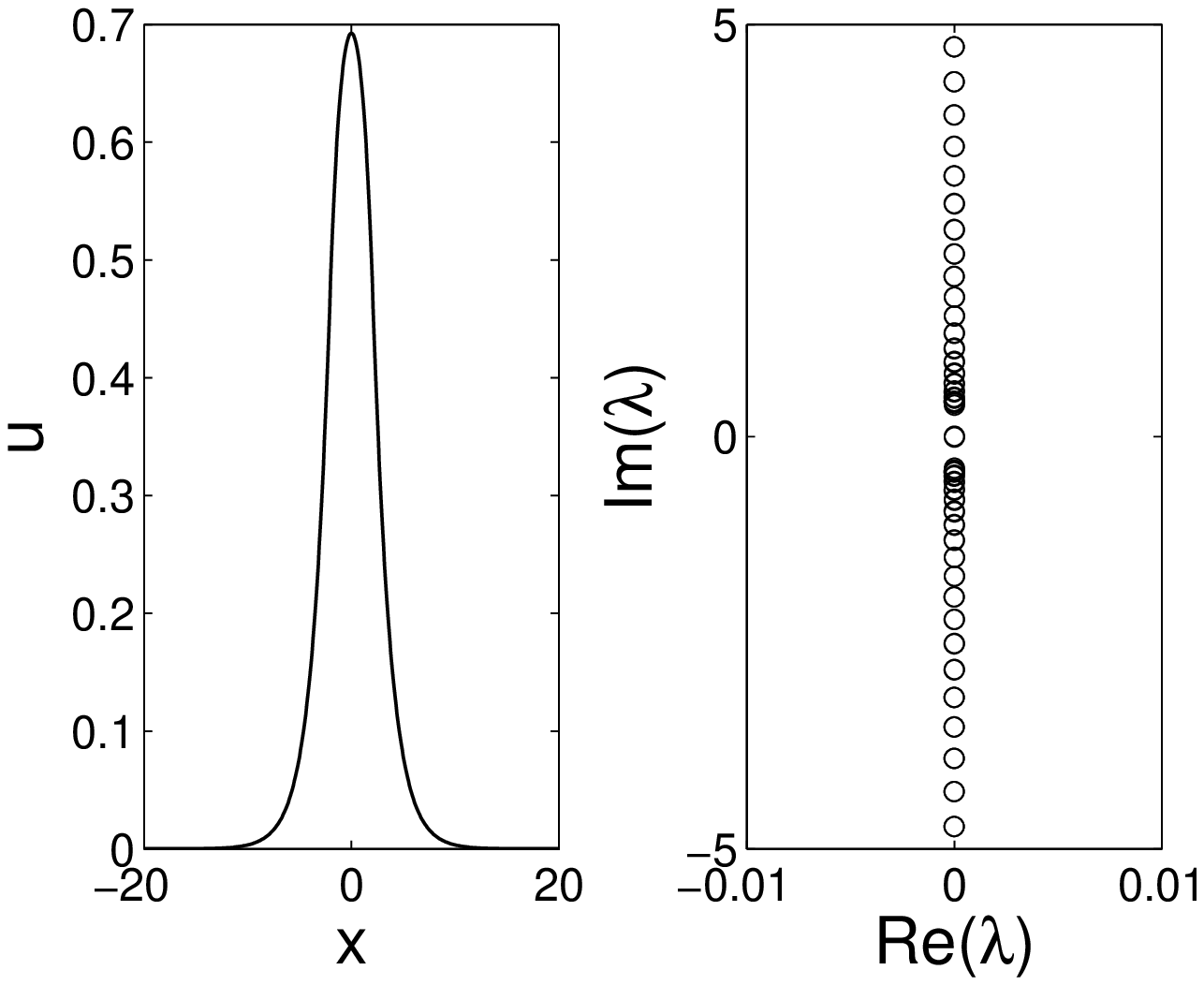}}
\subfigure[]{\includegraphics[scale=0.45]{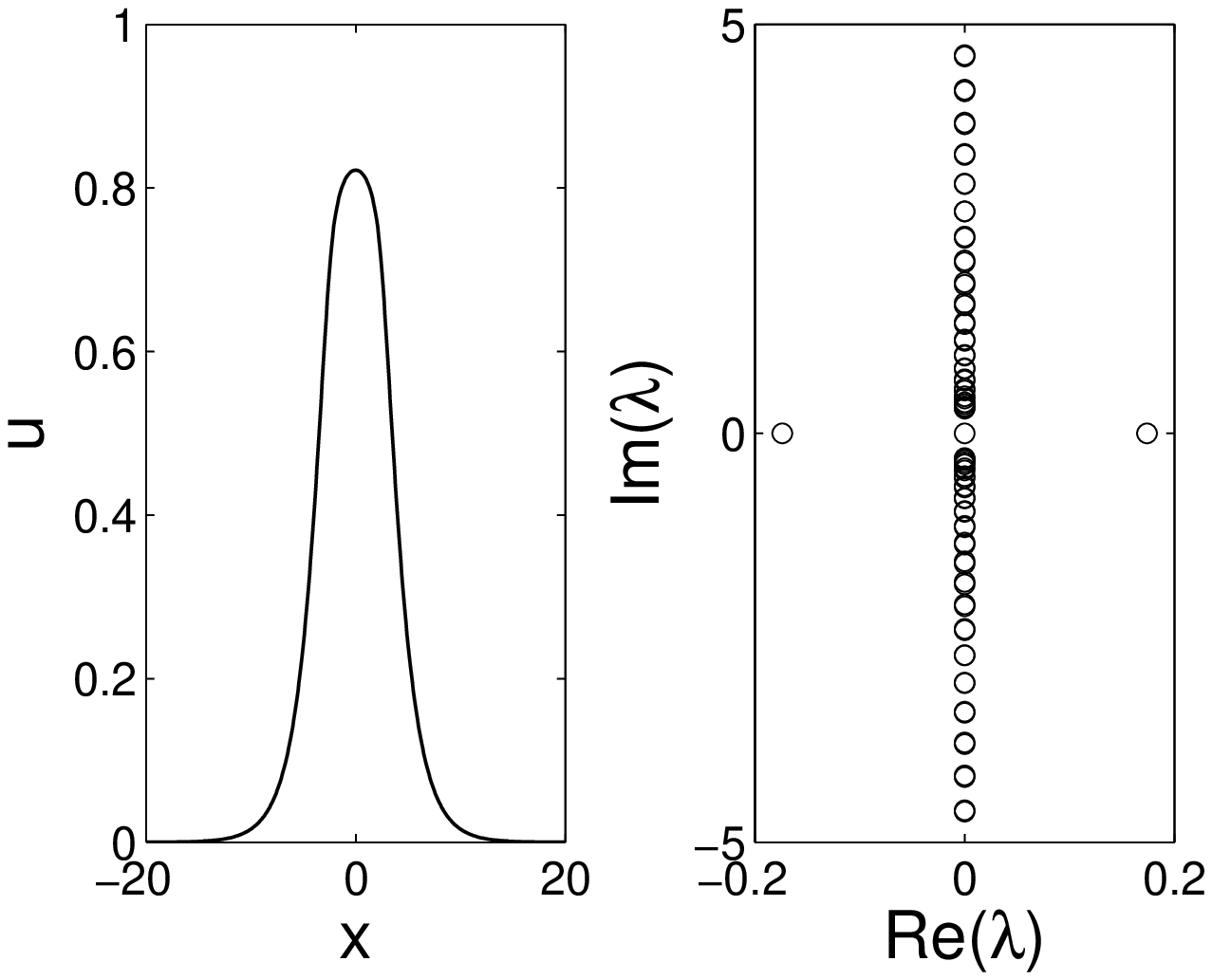}}
\subfigure[]{\includegraphics[scale=0.45]{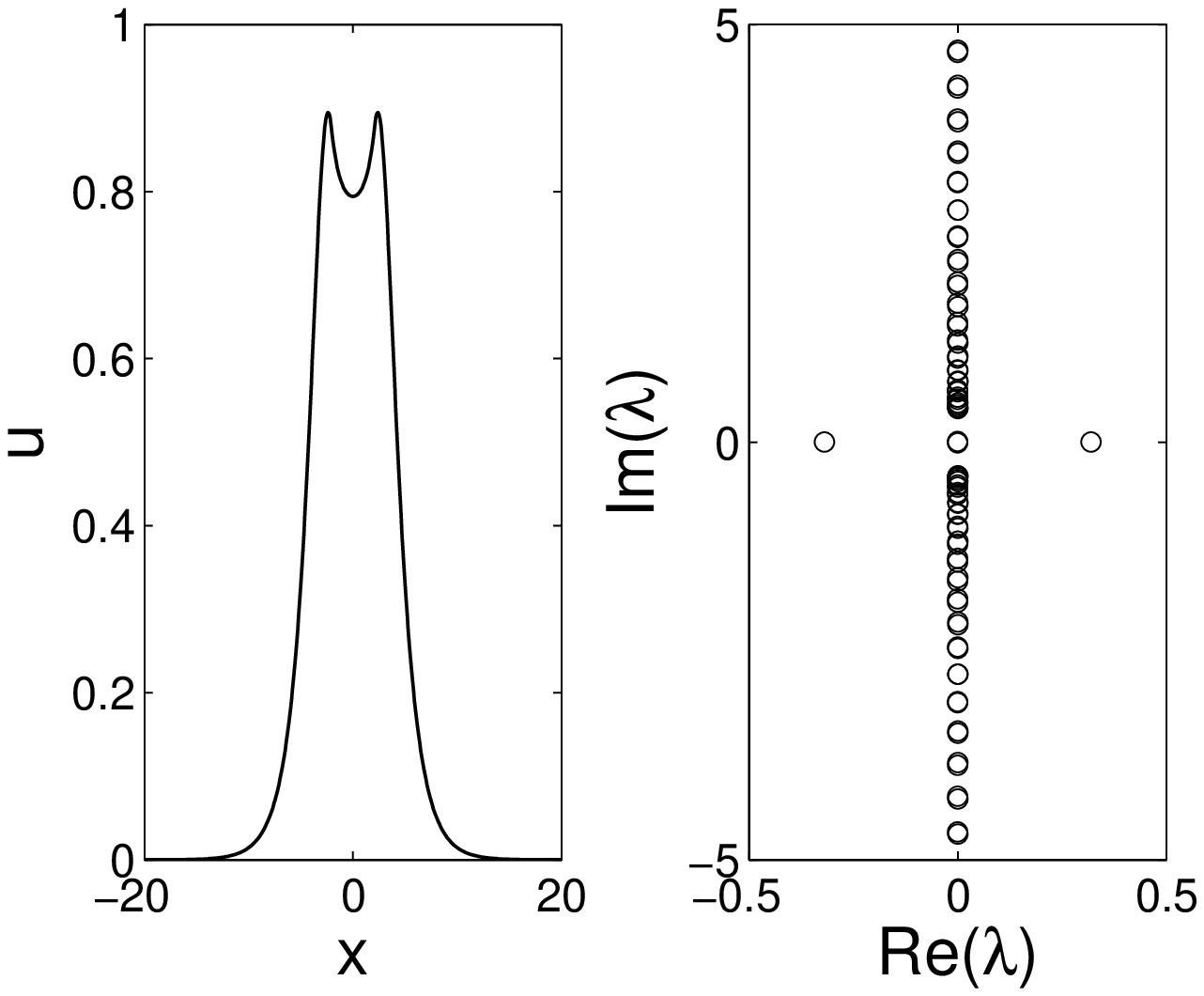}}
\subfigure[]{\includegraphics[scale=0.45]{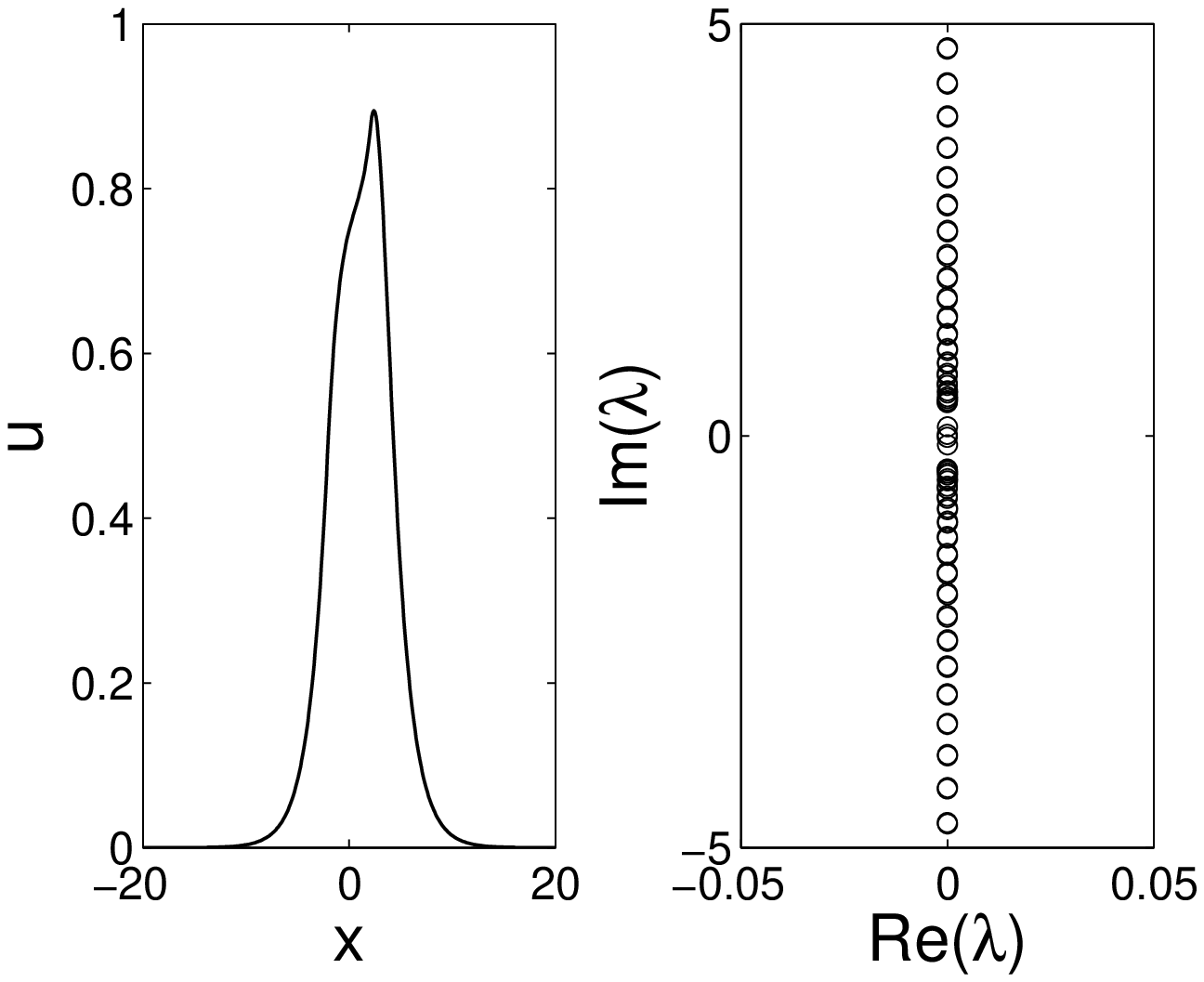}}
\caption{Profile of solutions at points indicated as A--D in Fig.\ \ref{fig2_num} and their eigenvalue structures in the complex plane. The phase-portraits of the solutions in each panel are presented in Fig.\ \ref{fig1_num_pp}.}
\label{fig1_num}
\end{center}
\end{figure}

As the parameter $\omega$ is increased further, there is a bifurcation at which the existence curve reverses direction. Symmetric solutions are unstable along this branch. Two solutions and their eigenvalues in the complex plane corresponding to point B and C are shown in Figs.\ \ref{fig1_num}(b) and (c), respectively. One can note that the instability of the solutions are due to the presence of a pair of eigenvalues with nonzero real part.

Interestingly, in addition to the symmetric states, at the bifurcation point where symmetric states become unstable, there is an existence curve emerging, corresponding to some numerically stable asymmetric states. A solution indicated as point D in Fig.\ \ref{fig2_num} is presented in Fig.\ \ref{fig1_num}(d).

\begin{figure}[tbhp!]
\begin{center}
\subfigure[]{\includegraphics[scale=0.5]{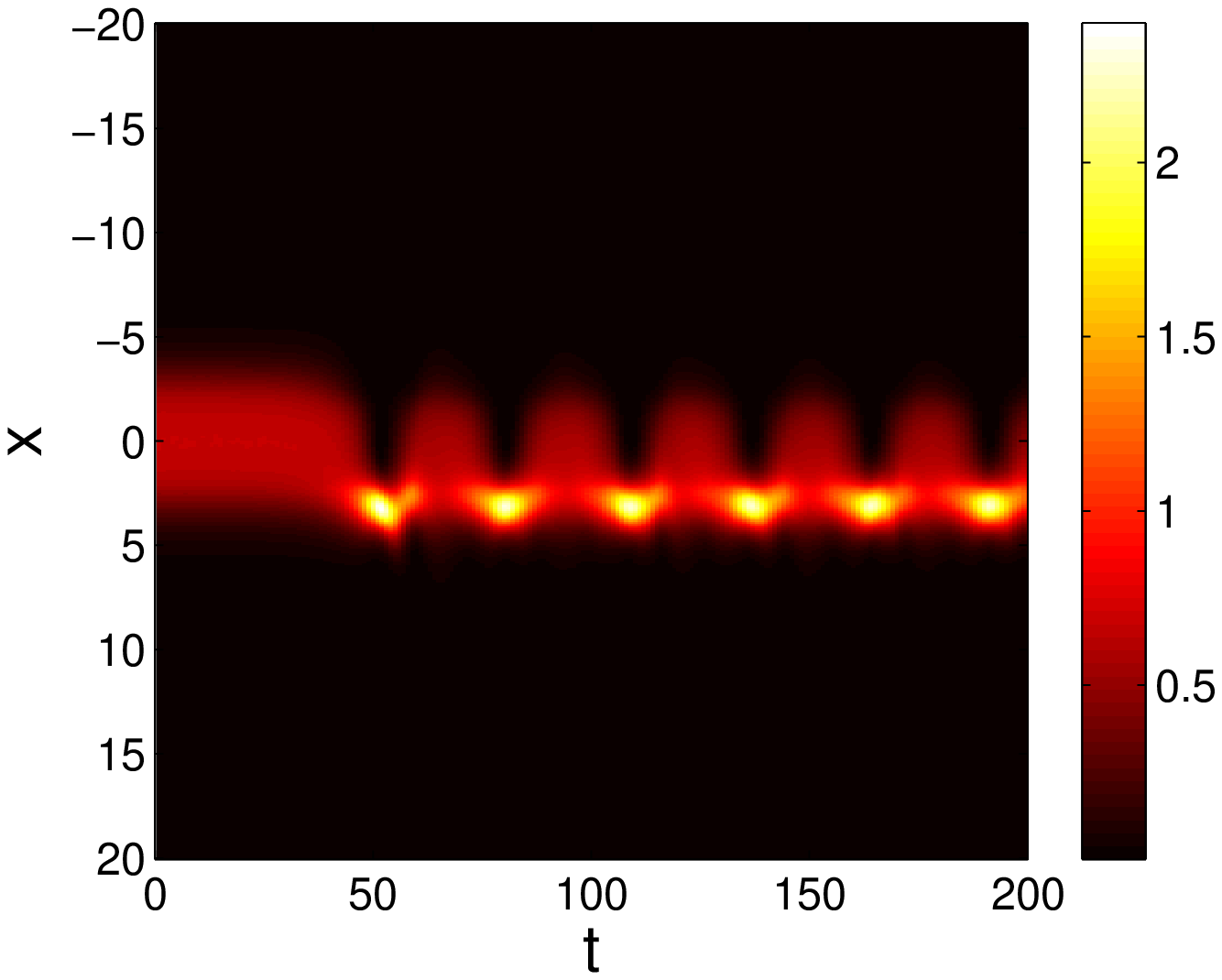}}
\subfigure[]{\includegraphics[scale=0.5]{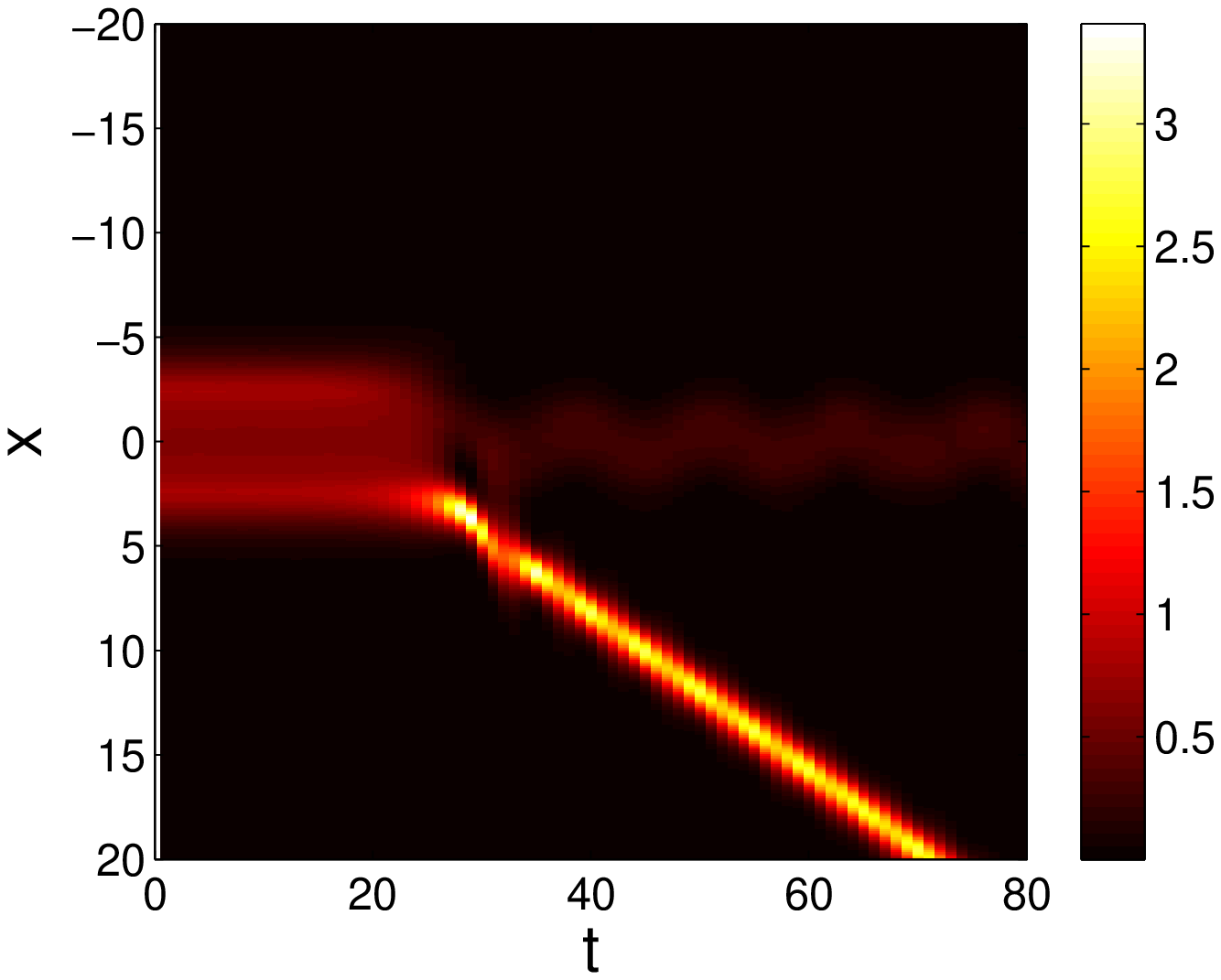}}
\caption{Time dynamics of solutions at points indicated as B and C in Fig.\ \ref{fig2_num}. Shown is the top view of $|\Psi(x,t)|^2$.}
\label{fig1_num_tevolv}
\end{center}
\end{figure}

Figure \ref{fig2_num} is similar to Fig.\ 1 in \cite{tran92b}, which is for the case of relatively small $\eta>0$. Writing $(V-\omega)\to W^2$ and $-\omega\to(W^2-V^2)$ and $L=1$, the time independent system of (\ref{gov2}) becomes the same as Eq.\ (1) in \cite{tran92b}. One important difference is that in our case, asymmetric solutions are all (at least numerically) stable in their existence domain.

When a solution is unstable, it is certainly of interest to see the dynamics near it. Here, we have solved the time-dependent governing equation (\ref{gov1}) using a Runge-Kutta method. Depicted in Figs.\ \ref{fig1_num_tevolv}(a) and (b) are the dynamics of solution (b) and (c) in Fig.\ \ref{fig1_num}, respectively, perturbed initially by small random disturbances. Shown is the modulus $|\psi(x,t)|^2$. One can clearly see that the instability of solution in Fig.\ \ref{fig1_num}(b) manifests in the form of spontaneous symmetry breaking, while the instability of the solution in Fig.\ \ref{fig1_num}(c) is in the form of a soliton generation, similar to that reported before in \cite{jleon04}.

\subsection{First excited state}

We have considered as well solutions bifurcating from the first excited corresponding to the second branch in Fig.\ \ref{transfig}. For the same $V$ and $L$ as above, this state bifurcates from the point $\omega\approx0.898$. We depict in Fig.\ \ref{fig3_num} the continuation of this state as $\omega$ varies.

\begin{figure}[tbhp!]
\begin{center}
\includegraphics[scale=0.5]{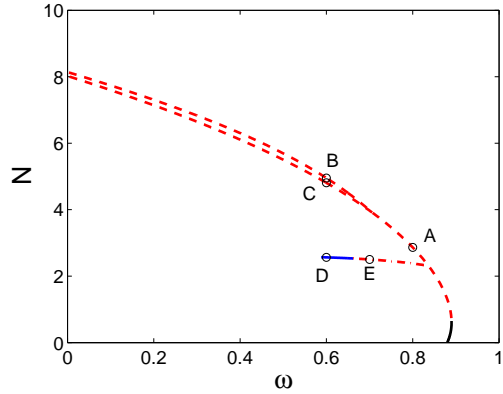}
\caption{The same as Fig.\ \ref{fig2_num}, but for the first excited state, corresponding to the second branch in Fig.\ \ref{transfig}. The bifurcation point of the linear state is $\omega\approx0.898$.}
\label{fig3_num}
\end{center}
\end{figure}

As the parameter $\omega$ increases from the bifurcation point, one will obtain a numerically stable symmetric state. When the parameter is increased further, there will also be a `direction reversal' point, where the symmetric state becomes unstable, similar to the case of fundamental state solutions above. Shown in Fig.\ \ref{fig4_num}(a) is an example of this state and its spectrum in the complex plane.

\begin{figure}[tbhp!]
\begin{center}
\subfigure[]{\includegraphics[scale=0.45]{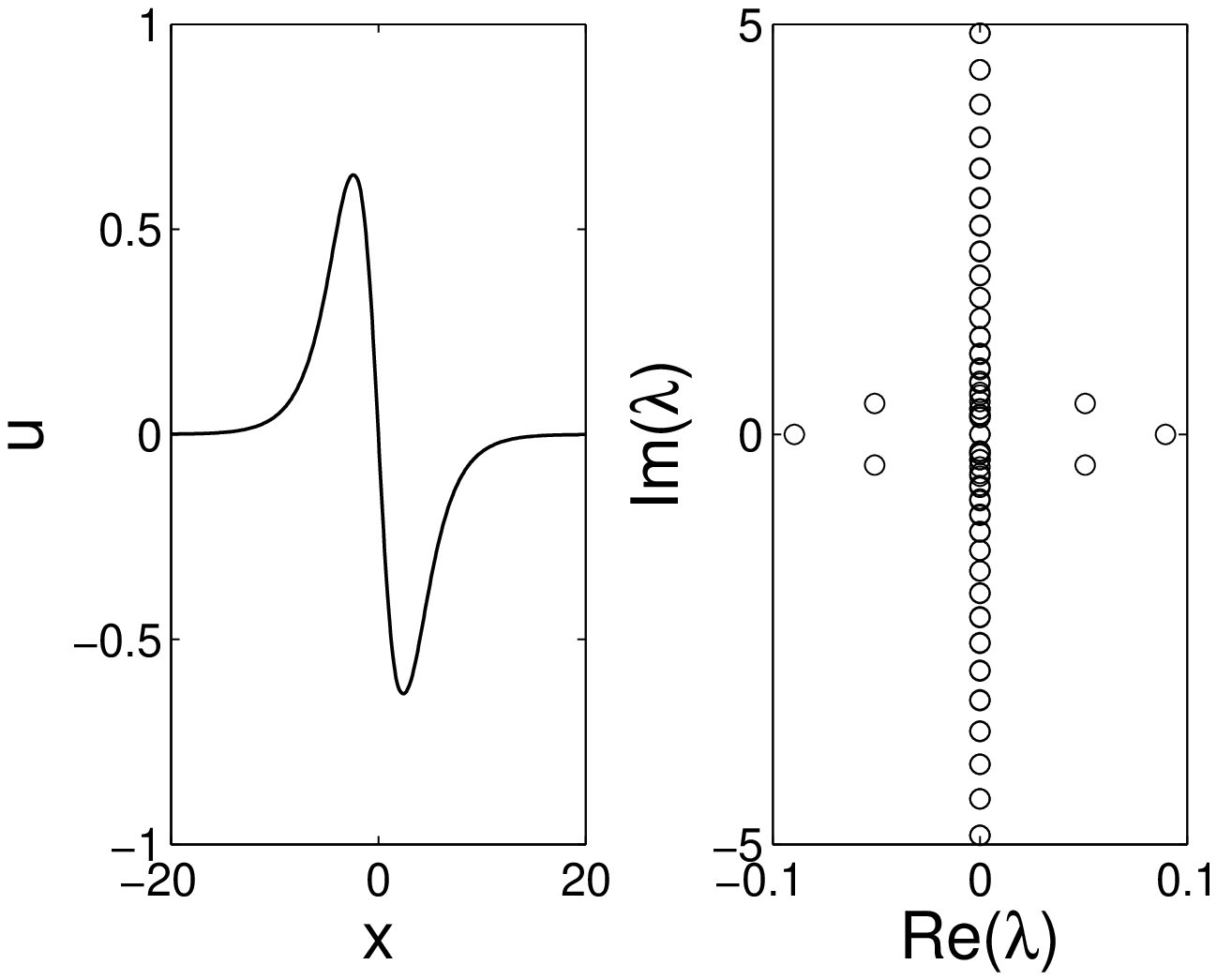}}
\subfigure[]{\includegraphics[scale=0.45]{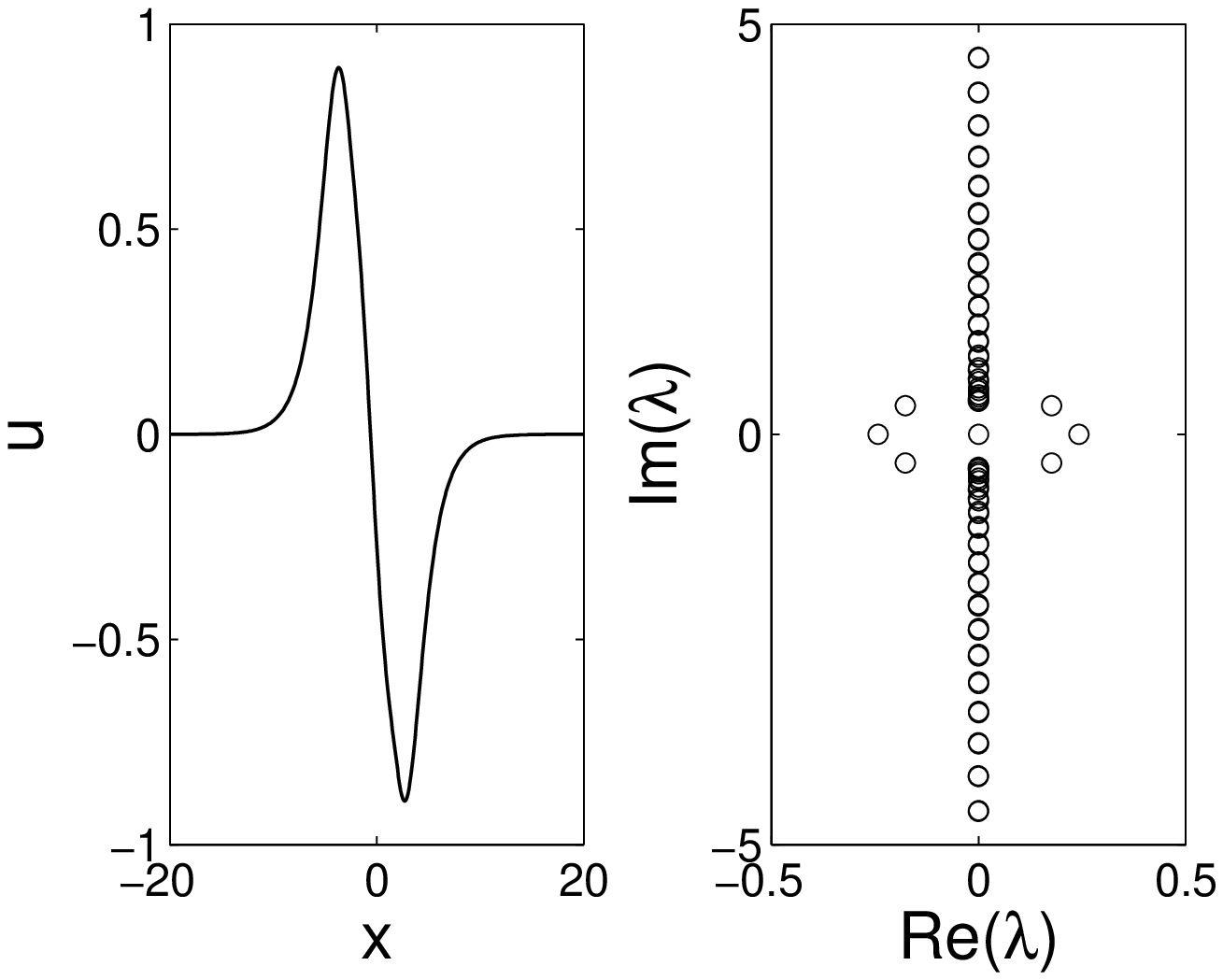}}
\subfigure[]{\includegraphics[scale=0.45]{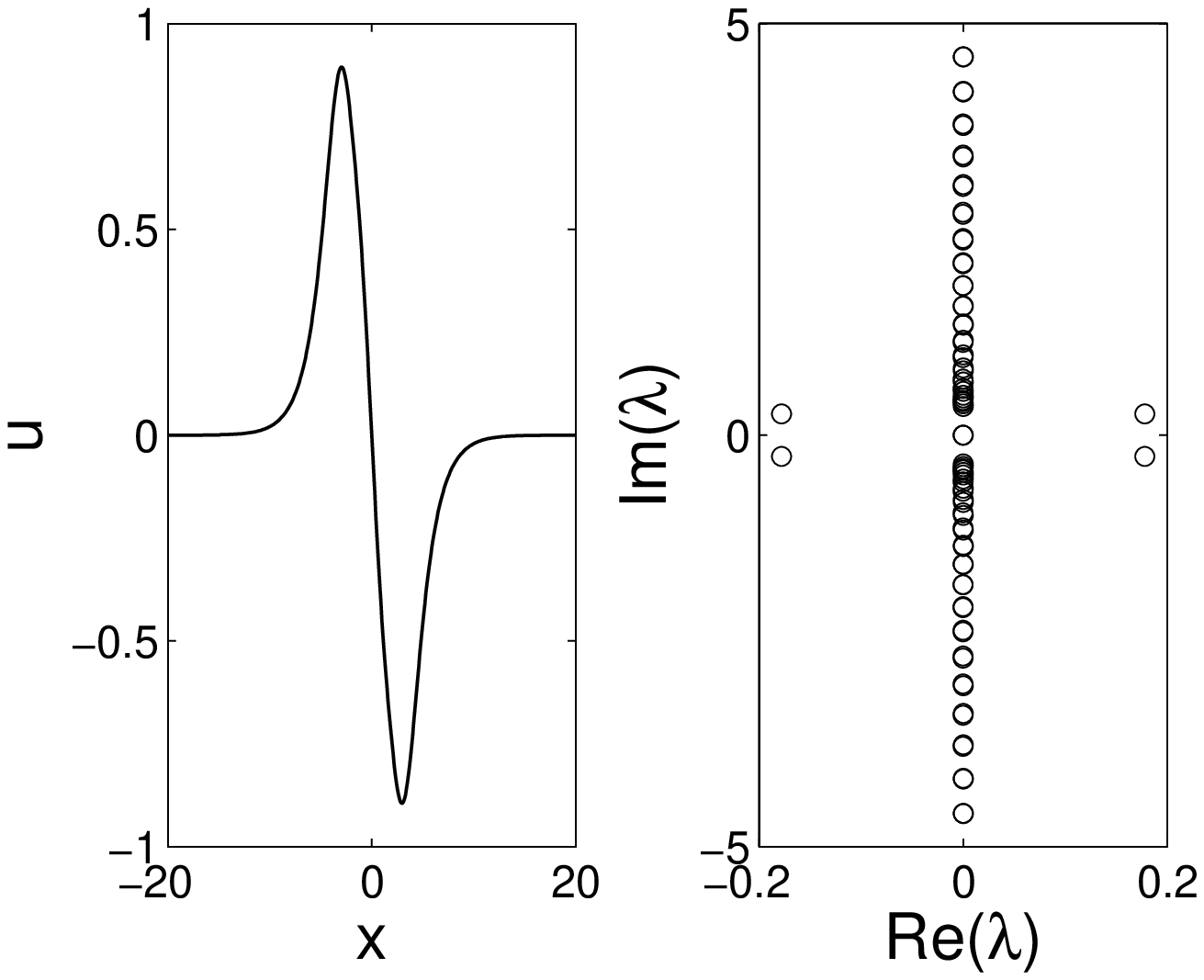}}
\subfigure[]{\includegraphics[scale=0.45]{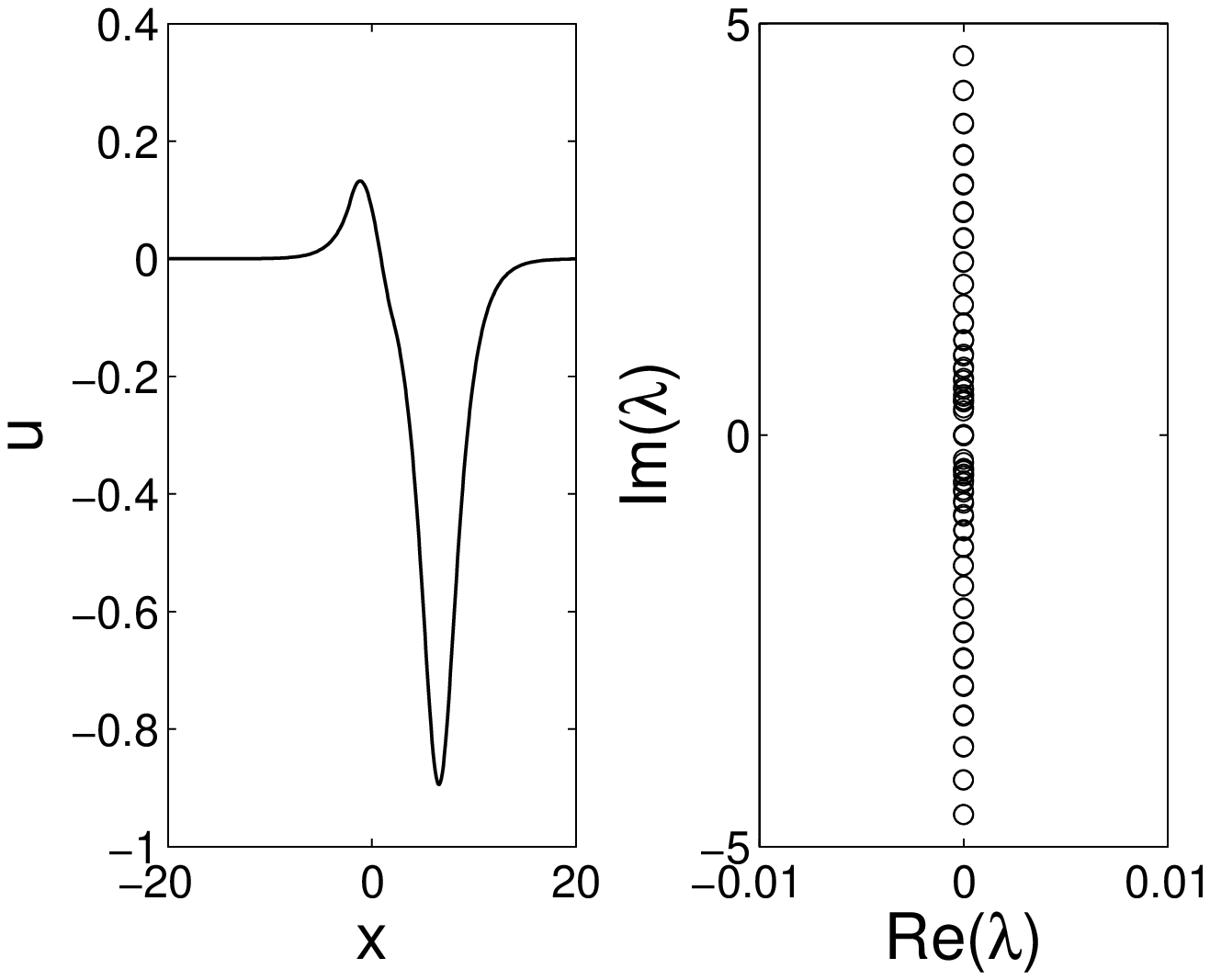}}
\subfigure[]{\includegraphics[scale=0.45]{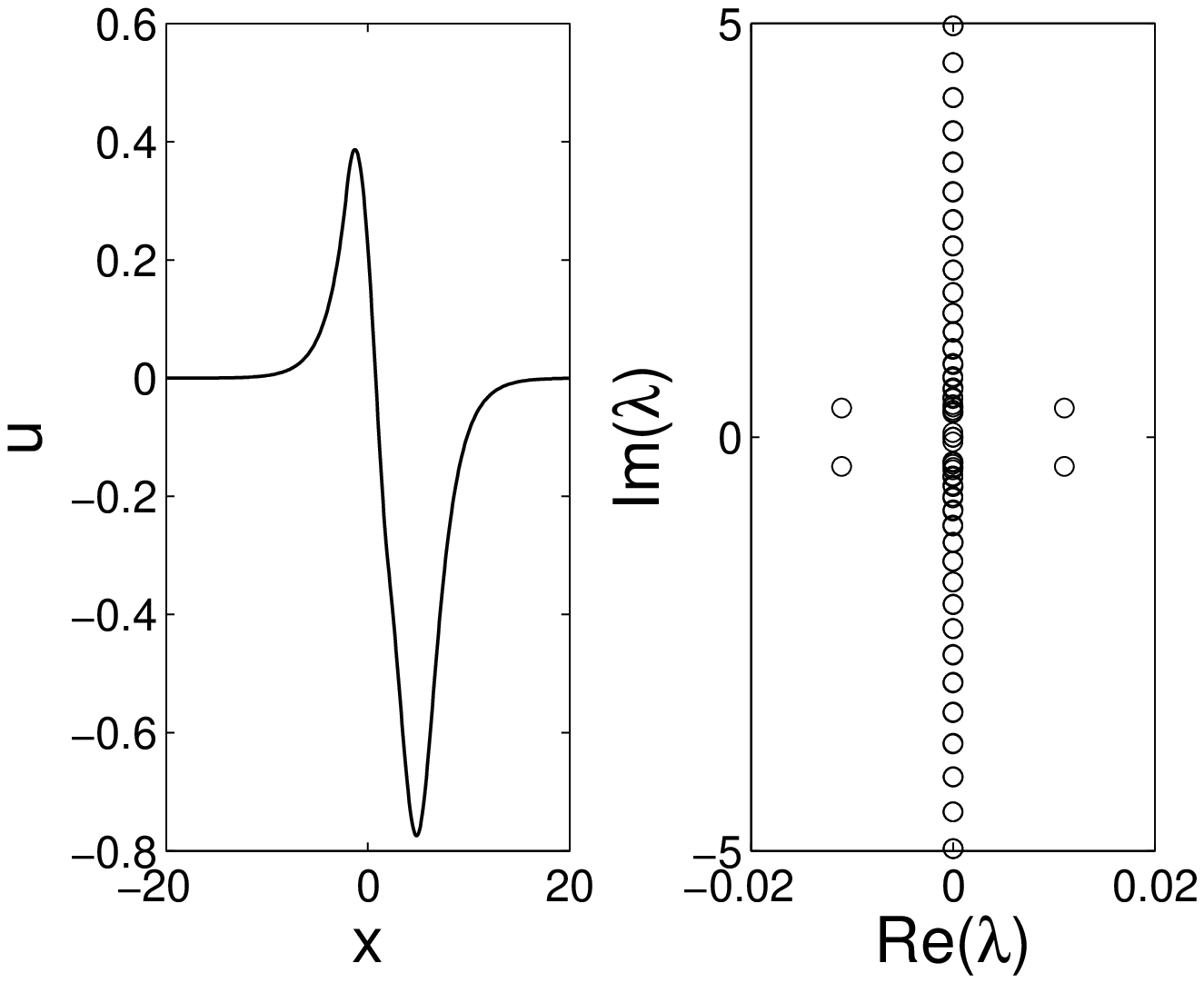}}
\caption{Profile of solutions and their corresponding eigenvalue structure in the complex plane at points indicated as A--D in Fig.\ \ref{fig3_num}. The solution phase-portraits are shown in Fig.\ \ref{fig4_num_pp}.}
\label{fig4_num}
\end{center}
\end{figure}

Besides similarities with the previous case, we also observed several differences here. These include the fact that there are now more than one existence branches corresponding to asymmetric solutions. Presented in Fig.\ \ref{fig4_num}(b) is a solution along the first asymmetric branch, indicated as point B in Fig.\ \ref{fig3_num}, and its eigenvalue structure. As a comparison, we also plot in Fig.\ \ref{fig4_num}(c) and (d) the symmetric and asymmetric solution from point C and D in Fig.\ \ref{fig3_num}, respectively, and their eigenvalues in the complex plane.

The two asymmetric solutions from point B and D above are clearly different. By viewing this first excited state as composed of two static out-of-phase solitons, the asymmetric solution B can be seen as composed of two solitons with one of them spatially displaced, while the solution D can be viewed as composed of two solitons with different amplitude.

It is also interesting to note that the asymmetric solution D is not always numerically unstable in its existence region. Numerically, we observed a region of stability and of instability for this asymmetric state, depicted as a solid and dashed line respectively in Fig.\ \ref{fig3_num}. In Fig.\ \ref{fig4_num}(e), we present an unstable asymmetric state and its spectra, where one can see the presence of two pairs of eigenvalues with nonzero real part.

\begin{figure}[tbhp!]
\begin{center}
\subfigure[]{\includegraphics[scale=0.5]{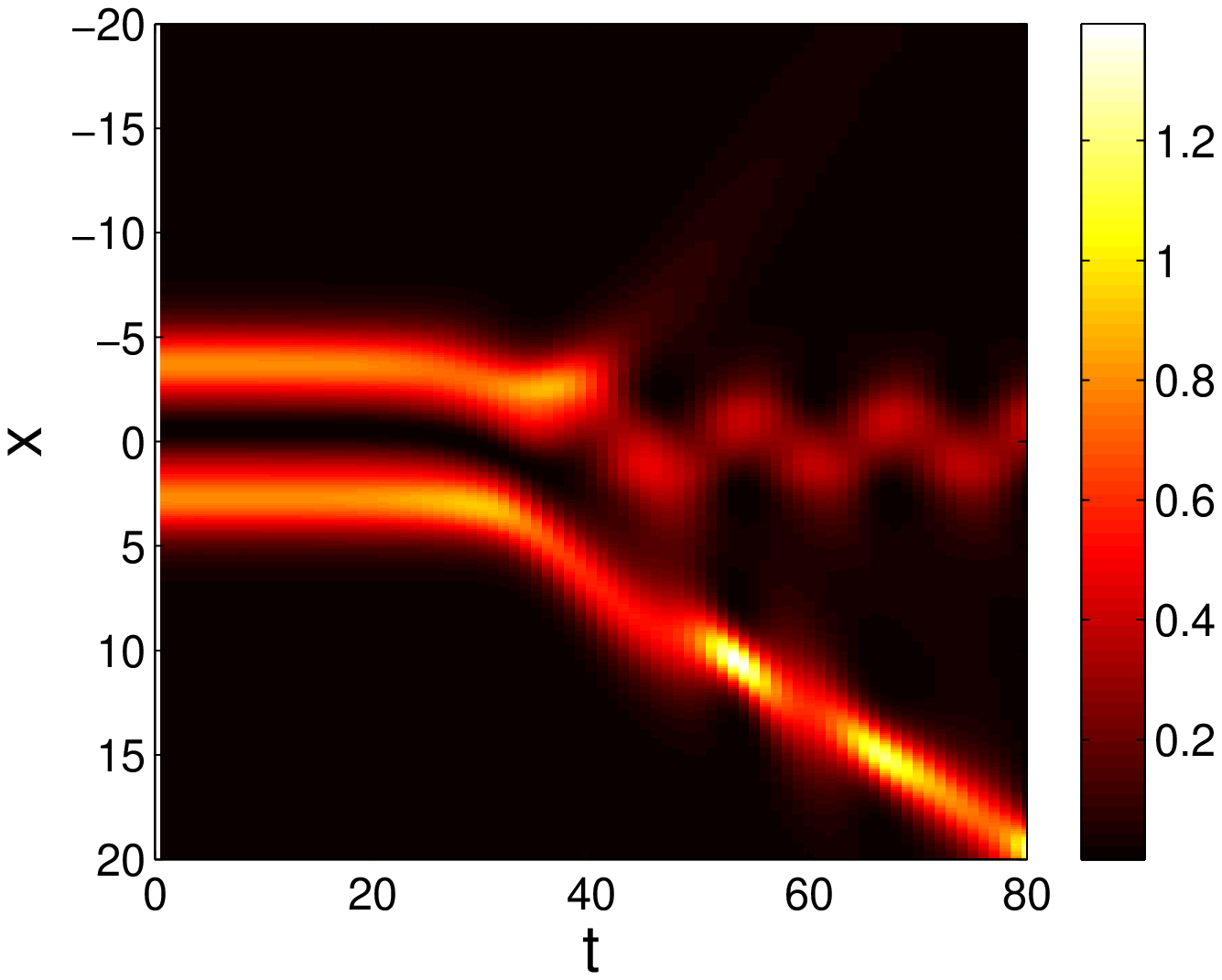}}
\subfigure[]{\includegraphics[scale=0.5]{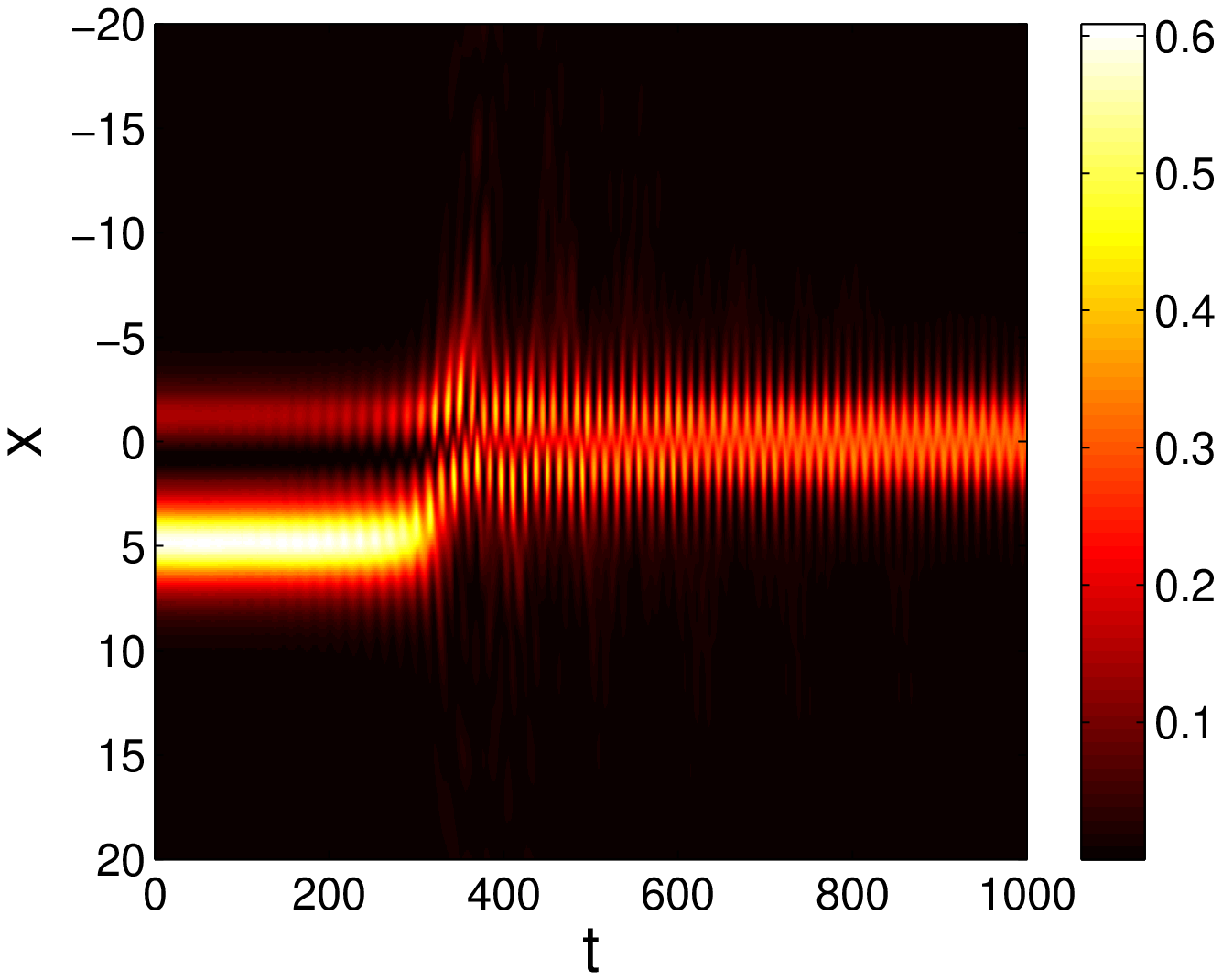}}
\caption{The dynamics in time of solutions at points indicated as B and E in Fig.\ \ref{fig3_num}, respectively. Shown is the top view of $|\Psi(x,t)|^2$.}
\label{fig4_num_tevolv}
\end{center}
\end{figure}

Typical time dynamics of unstable solutions for this case is presented in Fig.\ \ref{fig4_num_tevolv}. In particular, we plot the dynamics of asymmetric solution B and E under random perturbations. The dynamics in time of solution C is similar to that of solution B.

\section{Unstable symmetric solutions past the bifurcation point}
\label{sympos}

In the following, we will analytically prove the instability of symmetric solutions past the `direction reversal' point above. To show instability of the standing waves, we will show that $M$ has a real positive eigenvalue. This is done by applying the main theorem of  \cite{ckrtj88}. It can be shown that the following quantities are well defined (see for example \cite{ckrtj88}):
\begin{eqnarray*}
P &=& \textrm{ the number of positive eigenvalues of } D_{+} \\
Q &=& \textrm{ the number of positive eigenvalues of } D_{-}.
\end{eqnarray*}
\noindent We then have the following:
\begin{theorem}[\cite{ckrtj88}] If $P-Q \neq 0, 1,$ there is a real positive eigenvalue of the operator $M$.  \label{th:ckrtj88}
\end{theorem}
From Sturm-Liouville theory, $P$ and $Q$ can be determined by considering solutions of $D_+ v = 0$ and $D_- v = 0$, respectively. In fact, they are the number of
zeros of the associated solution $v$. Notice that $D_- v = 0$ is actually satisfied by the standing wave itself, and that $D_+ v = 0$ is the equation of variations of the
standing wave equation. It follows that:
\begin{eqnarray*}
Q &=& \textrm{ the number of zeros of the standing wave } u(x). \\
P & = & \textrm{ the number of zeros of a solution to the variational equation along $u(x)$. }
\end{eqnarray*}
We will focus first on the case when we have a positive localized steady state solution, i.e.\ $Q = 0$.

\subsection{Positive solutions}

The idea is to use a dynamical systems point of view, and
geometric properties of the solution curves in the phase portrait to establish
when $P \geq 2$. The $t$ independent solutions to equation (\ref{gov2}) can be represented by composite phase portraits constructed by a superpositioning of the
phase portraits of the `outer' system:
\begin{equation}
u_x = y, \qquad
y_x = (V-\omega)u-u^3,
\label{eq:outer}
\end{equation}
and the `inner' one:
\begin{equation}
u_x = y, \qquad
y_x = -\omega u+u^3.
\label{eq:inner}
\end{equation}
We can view the composite picture as a single, non-autonomous system with phase plane given by:
\begin{equation}
\begin{array}{lll}
u_x = y, \\
y_x = \left\{\begin{array}{lll}
   (V-\omega)u-u^3, & |x|>L, \\
   -\omega u+u^3, &|x|<L.
  \end{array}\right.
\end{array}
\label{fiberdyn}
\end{equation}

In the phase plane of (\ref{eq:outer}), the outer system admits a soliton solution, given by the equation:
\begin{equation} \label{eq:homo}
y^2 = (V - \w) u^2 - \frac{u^4}{2},
\end{equation}
\noindent while the inner system (\ref{eq:inner}) admits a heteroclinic orbit in the phase plane given by:
\begin{equation} \label{eq:hetero}
y^2 = - \w u^2 - \frac{u^4}{2} + \frac{\w^2}{2}.
\end{equation}
The conditions that $u$ and $u_x$ decay to zero as $x$ goes to $\pm \infty$, mean that in the superposed phase
portraits, all steady-state solutions that we are interested in will lie on the soliton curve as $x \to \pm \infty$.

Solution curves of the inner system are given by:
\begin{equation} \label{eq:innerC}
y^2 = - \w u^2 + \frac{u^4}{2} + C.
\end{equation}
The solutions we are interested in will travel in the phase plane out from the origin along the homoclinic orbit of the outer system described by (\ref
{eq:homo}) and then `flip' to the inner system for $2L$ units of $x$, and then `flip' back to the outer system along the homoclinic orbit. Define $(u_0,y_0)$ as the
point in the phase plane of (\ref{fiberdyn}) where the solution initially flips from the outer to the inner system, and define $(u_1,y_1)$ as the point in the phase plane
where the solution returns to the outer system. Using this notation we can write the equation of the part of the solution curve in the inner system in the (composite) phase plane as:
\begin{equation} \label{eq:inneru0}
y^2 = - \w u^2 + \frac{u^4}{2} + Vu_0^2-u_0^4.
\end{equation}

As the parameters, $V$ and $\w$ vary, the relative position of the homoclinic orbit described by (\ref{eq:homo}) and the heteroclinic orbit described in (\ref
{eq:hetero}) will change. We consider the qualitative and numerical differences that occur for various values of $L$ corresponding to different
configurations of the potential $V$ and the frequency $\w$. Namely when the curves described by (\ref{eq:hetero}) and (\ref{eq:homo}) are `close'  together, linearly
unstable standing wave solutions to (\ref{gov1}) appear.

\begin{theorem} \label{th:main}
Unstable positive localized solutions to (\ref{gov1}) occur whenever $\frac{\w}{V}<\frac{3}{4} $. The unstable symmetric solutions are the ones which are symmetric in the phase plane with respect to the $u$-axis and leave the homoclinic orbit at a point $(u_0, y_0)$ satisfying $\sqrt{\frac{V}{2}} < u_0$.
\end{theorem}

The theorem will be proved by showing that $P \geq 2$ for such a solution. From Sturm-Liouville theory we have that $P$ will be the number of zeros of a solution to
the variational equation $D_+ v = 0$ satisfying the boundary conditions $ v
\to 0$ as $x \to - \infty$.  A solution to the variational equation associated with a solution of (\ref{stat1}) can be found by following a tangent vector around the orbit
under the flow of the linear equations (\ref{fiberdyn}). Satisfying the initial conditions means the solution will be a tangent vector to the orbit of the solution of (\ref{stat1}) in the phase plane until the
discontinuity of (\ref{stat1}). The number of zeros of such a solution can be found by determining the number of times that such a vector must pass through the vertical
as the base point ranges over the entire orbit.

Denote by $\textbf{b}^O(u,y) = (\textbf{b}^O_1(u,y),\textbf{b}^O_2(u,y)) $ a tangent vector to the homoclinic orbit of
the outer system at the point $(u,y)$, and let $\textbf{b}^I(u,y) = (\textbf{b}^I_1(u,y),\textbf{b}^I_2(u,y))$, be a
tangent vector to an orbit of the inner system at the point $(u,y)$. As $(u,y)$ varies along the
homoclinic orbit, $\textbf{b}^O(u,y) = (p,q)$ solves the linear system:
\begin{equation}
\begin{array}{lll}
p_x= q, \\
q_x = (V-\omega)p-3u^2p,
\end{array}
\label{varouter}
\end{equation}
\noindent while $\textbf{b}^I(u,y)$ solves the linear system:
\begin{equation}
\begin{array}{lll}
p_x = q, \\
q_x =
   -\omega p+3u^2p,
\end{array}
\label{varinner}
\end{equation}
\noindent as $(u,y)$ travels along a curve given by (\ref{eq:innerC})
The composite linear variational equation is then given by
\begin{equation}
\begin{array}{lll}
p_x = q, \\
q_x = \left\{\begin{array}{lll}
   (V-\omega)p-3u^2p, & |x|>L, \\
   -\omega p+3u^2p, &|x|<L.
  \end{array}\right.
\end{array}
\label{vardyn}
\end{equation}
We denote the explicit vectors, $\textbf{b}^O(u_0,y_0)$ and
$\textbf{b}^I(u_1,y_1)$, by $\beta^O$ and $\beta^I$ respectively. Let $F^O$ and $F^I$ denote the flow of the outer (\ref{eq:outer}) and inner (\ref{eq:inner})
systems respectively and $\Phi^O$ and $\Phi^I$ denote the variational flow of the outer and inner systems along solutions. If $\textbf{d}(u,y) = (\textbf{d}_1(u,y),\textbf{d}_2(u,y))$ is a vector in the tangent space of the
phase plane at the point $(u,y)$, and if $(u,y)$ is flowed under $F^O$ and $F^I$ to $F^O(u,y)$ and $F^I(u,y)$
respectively, let $\Phi^O(\textbf{d})(F^O(u,y))$, $\Phi^I(\textbf{d})(F^I(u,y))$ denote the image of the vector
$\textbf{d}$ under the flow of the outer and inner variational systems respectively.

Before moving to the proof of the theorem we need two geometric facts about the flows $\Phi^O$ and $\Phi^I$.
\begin{fact} \label{fact:pres} The variational flows $\Phi^O$ and $\Phi^I$ map the tangent line to a solution to (\ref{eq:outer}) and (\ref{eq:inner}) at a point $(u,y)$, to
the tangent space of the solution at the point $F^O(u,y)$ and $F^I(u,y)$ respectively.
\end{fact}
\begin{fact}\label{fact:ori}
The flows $\Phi^O$ and $\Phi^I$ are orientation preserving.
\end{fact}

This second fact is reflected in the following way; if $\textbf{b}(u,y) = (b_1(u,y),b_2(u,y))$ and $\textbf{d}(u,y) = (d_1(u,y), d_2(u,y))$ are two vectors tangent to the phase space at
$(u,y)$, then the sign of the cross product of the two vectors is unchanged under the flows. That is:
\begin{equation}\label{eq:ori}
\textrm{sgn} ( \textbf{b}(u,y) \times \textbf{d}(u,y) ) = \textrm{sgn} ( \Phi^O( \textbf{b})(F^O(u,y)) \times \Phi_{O} (\textbf{d})(F^O(u,y)))
\end{equation}
\noindent For every point in the orbit of the flow $F^O$. The above also holds with the outer flow replaced by the inner flow.

A consequence of these two facts for the system of interest is the following lemma:
\begin{lemma}\label{lem:asym} If $\textbf{d}(u,y)$ is a vector in the tangent space to the phase plane of the homoclinic orbit at the point $(u,y)$, with $u>0$ such that
$\textbf{d}(u,y) \times \textbf{b}^O(u,y) > 0 $, then
\begin{equation}
\lim_{x \to \infty} \Phi^O(\frac{\textbf{d}}{|\textbf{d}|})(F^O(u,y)) = k \left(\begin{array}{ccc} 1 \\ \sqrt{V-\w } \end{array}\right)
\end{equation}
\noindent where $k$ is a positive real number.
\end{lemma}
\begin{proof}Because $\textbf{b}^O(u,y)$ is a tangent vector to the homoclinic orbit, and $u>0$,
\begin{equation}
\lim_{x \to \infty} \Phi^O(\frac{\textbf{b}^O)}{|\textbf{b}^O|})(F^O(u,y)) = \lim_{x \to \infty} \frac{\textbf{b}^O}{|\textbf{b}^O|}(F^O(u,y)) = \frac{1}{\sqrt{1+V-\w}}\left(\begin{array}{ccc} -1 \\ \sqrt{V-\w } \end{array}\right)
\end{equation}
\noindent and since $\textbf{d}$ is not tangent to the homoclinic orbit,
\begin{equation}
\lim_{x \to \infty} \Phi^O(\frac{\textbf{d}}{|\textbf{d}|})(F^O(u,y)) = k \left(\begin{array}{ccc} 1 \\ \sqrt{V-\w } \end{array}\right),
\end{equation}
\noindent where $k$ is some real number. This is because $\left(\begin{array}{ccc} 1 \\ \sqrt{V-\w } \end{array}\right)$ is the
unstable asymptotic eigenvector of (\ref{varouter}). But from the first and second fact we have that
\begin{equation}
0 < \lim_{x\to \infty} \Phi^O(\frac{\textbf{d}}{|\textbf{d}|})(F^O(u,y)) \times \frac{\textbf{b}}{|\textbf{b}|}(F^O(u,y)) = \frac{2k\sqrt{V-\w}}{\sqrt{1+V-\w }},
\end{equation}
\noindent which means that $k$ must be positive.
\end{proof}

In order to determine the number of zeroes of a solution to $D_+ v = 0$, and $v \to 0$ as $x \to -\infty$,  we need to
determine the number of times a vector gets pushed through the vertical by the flow as its base point moves along the
orbit. In order to count this, we break the orbit up into three parts. These are when $x<-L$, $|x|<L$ and $x>L$. Let
$A_1$ denote the number of zeros of the solution to the variational equation as the basepoint corresponds to the range
$x<-L$, $A_2$ the number for $|x|<L$, and $A_3$ for $x>L$. We have the following:
\begin{itemize}
\item The number $A_1$ is the number of times $\textbf{b}^O(u,y)$, a vector tangent to the homoclinic orbit passes
through the vertical as $(u,y)$ travels from $(0,0)$ along the homoclinic orbit of the outer system to the point
$(u_0,y_0)$.
\item The number $A_2$ is the number of times $\Phi^I(\beta^O)(u,y)$ passes through the vertical as $(u,y)$ travels
from $(u_0,y_0)$ along the path $y^2 = \frac{u^4}{2} - \w u^2 + V u_0^2 - u_0^4$ to the point $(u_1,y_1)$.
\item The number $A_3$ is the number of times $\Phi^O(\Phi^I(\beta^O)(u_1,y_1))(u,y)$ passes through the vertical as
$(u,y)$ travels along the homoclinic orbit of the outer system from $(u_1,y_1)$ to $(0,0).$
\end{itemize}

The number of zeros of a solution of (\ref{vardyn}) will then be $ P = A_1+A_2+A_3$. Theorem \ref{th:ckrtj88} then says that
if $u(x)>0$ and if $A_1+A_2+A_3 \geq 2$, then the underlying orbit in the phase plane represents an unstable standing wave solution of (\ref{gov1}).
We are now ready to prove Theorem \ref{th:main}.
\begin{proof}[Proof of Theorem \ref{th:main}]
We have two cases to consider, the first is when $\sqrt{\frac{V}{2}}<u_0<\sqrt{\w}$ and
the second is when $u_0 > \sqrt{\w}$. In both cases however $u_0 > \sqrt{V-\w}$.

\begin{figure}[tbp!]
\begin{center}
\subfigure[]{\includegraphics[height= 2.in]{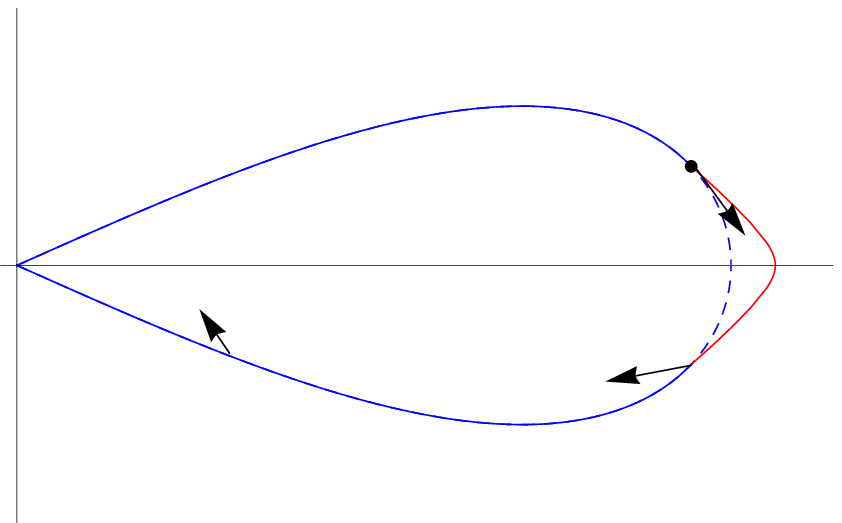}
\label{fig1:subfig1}}
\subfigure[]{\includegraphics[height= 2.in]{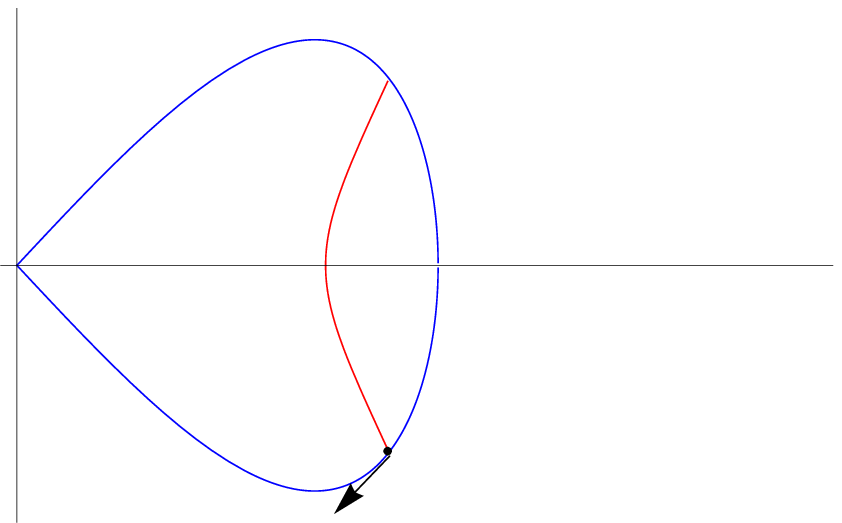}
\label{fig1:subfig2}}
\end{center}
\caption{(a) An unstable orbit of the first type. The dashed line is the homoclinic orbit of the outer system. The dot represents the jumping off point. The black arrows represent the image of a tangent vector to a solution under the variational flows $\Phi^{O,I}$ at various values. (b) An unstable orbit of the second type. The dot represents the jumping off point. The black arrow represents where $\Phi(\bf{b}^O)$ stops being tangent to the orbit of the solution.}
\end{figure}

\textbf{Case 1.} Here $\frac{V}{2}<u_0<\sqrt{\w}$. An example is as in figure \ref{fig1:subfig1}. Note that in this case $y_0> 0$. We show that $A_2 \geq 1$ and $A_2+A_3 \geq 2$, so the solution  represents an unstable standing wave.
First we show that $A_2 \geq 1$. Notice that $\beta^O = (y_0, (V-\w)u_0-u_0^3)$, and that the first coordinate is positive while the second coordinate of
$\beta^O$ is negative since $\frac{V}{2}<u_0<\sqrt{\w}$. Now notice that $\textbf{b}^I(u_0,y_0) = (y_0, u_0^3 - \w
u_0)$, and again the second coordinate $u_0 > \sqrt{V-\w}$. This means that $\beta^O \times \textbf{b}^I(u_0,y_0)  = y_0u_0(2u_0^2-V>
0$, because $u_0> \sqrt{\frac{V}{2}}$ and $y_0$ is positive. Now we flow both vectors along the inner system until we get to the point where the solution to the original equation crosses the $u$ axis in the phase plane of the system described by (\ref{fiberdyn}). Denote this point by $(u_{\textrm{max}},0)$. We then have the following:
\begin{equation*}
0 < \Phi^I(\beta^O)(u_{\textrm{max}},0) \times \textbf{b}^I(u_{\textrm{max}},0) = (u_{\textrm{max}}^3 - \w u_{\textrm
{max}})\Phi^I_1(\beta^O)(u_{\textrm{max}},0).
\end{equation*} Since $u_0<\sqrt{w}$ this means that $u_{\textrm{max}} = \sqrt{\w - \sqrt{\w^2 - 2(Vu_0^2 + u_0^4)}}< \sqrt{\w}$ and so the first coordinate of
$\Phi^I(\beta^O)(u_{\textrm{max}},0)$, $\Phi^I_1(\beta^O)(u_{\textrm{max}},0) < 0.$ But this means that the flow has
pushed the original tangent vector to the homoclinic orbit through the vertical at least once by this point, since the sign of the first coordinate has changed, so $A_2
\geq 1$.

To study $A_3$, we must determine the number of times that $\Phi^O(\Phi^I(\beta^O)(u_1,y_1))(u,y)$ passes through the vertical as the base point travels along the homoclinic orbit from $(u_1,y_1)$ to $(0,0)$.  We first apply the result of the lemma to the vector $\beta^I$. Observe that $\beta^I =
(-y_0,u_0^3-\w u_0)$, (since in this case $(u_1,y_1) = (u_0,-y_0)$) and that the first, and second coordinates of $\beta^I$ are negative. Now
$\beta^I \times \textbf{b}^O(u_1,y_1) = y_1(Vu_1-2u_1^3) > 0$, since $u_1 =u_0> \sqrt{\frac{{V}}{2}}$. So applying the
result of lemma \ref{lem:asym} means that there is a point on the homoclinic orbit $(u',y')$ say, where $\Phi^O(\beta^I)(u',y')$
is pointing vertically upward. Write $\Phi^O(\beta^I)(u',y') := (0, a)$, with $a$ positive. Now we apply the fact that the variational
flows are orientation preserving. We have
\begin{equation*}
0< \Phi^O(\Phi^I(\beta^O)(u_1,y_1))(u',y') \times (0,a) = a(\Phi^O_1(\Phi^I(\beta^O)(u_1,y_1))(u',y'))
\end{equation*}
\noindent which means that the first coordinate of $\Phi^O(\Phi^I(\beta^O)(u_1,y_1))(u',y')$ must be positive at this point. The first coordinate under the variational flow
$\Phi^I_1(\beta^O)$ was shown to be negative above, so this means that the vector $\Phi^I(\beta^O)$ must have either passed through the vertical once more,
or been pushed through the vertical by $\Phi^O$, by this point, so $A_2 + A_3 \geq 2$.

Thus we have that the number of zeros to the variational equation must be greater than or equal to two, so the corresponding orbit must represent an unstable standing wave.

\textbf{Case 2} In this instance $u_0 > \sqrt{\w}$. An example of the phase portrait of an orbit of this type is in
figure \ref{fig1:subfig2}. Here we remark that since $u_0 > \sqrt{\w}$ this means that $y_0$ must be negative. In
this case we show that $A_1 \geq 1$ and that $A_1+ A_2+ A_3 \geq 2$.

To see that $A_1 \geq 1$, we note that the tangent line to the homoclinic orbit is vertical when the homoclinic
orbit crosses the $u$ axis. As the base point moves along the homoclinic orbit from $(0,0)$ to $(u_0,y_0)$, it must cross the $u$ axis,  so $A_1 \geq 1$.

Now we study $A_3$. Call the point where the homoclinic orbit crosses the $u$-axis $(u_{\textrm{max}},0)$. Note that $\beta^O = (y_0, (V-w)u_0 - u_0^3)$ and that the sign of the first coordinate is negative. The sign of $\Phi^I(\beta^O)(u_1,y_1) \times \textbf{b}^O(u_1,y_1)$ can be either positive, negative or zero. If it is positive apply lemma \ref{lem:asym} to see that
\begin{equation*}
\lim_{x \to \infty} \frac{\Phi^O(\Phi^I(\beta^O)(u_1,y_1))(u,y)}{|\Phi^O(\Phi^I(\beta^O)(u_1,y_1))(u,y)|} = \left(\begin{array}{ccc} k \\ k \sqrt{V-\w } \end{array}\right),
\end{equation*}
\noindent where $k$ is positive, so the flow must have pushed the vector through the vertical once more, since the first coordinate in the limiting vector is positive. Thus $A_1+A_2+A_3 \geq 2$.

If $\Phi^I(\beta^O)(u_1,y_1) \times \textbf{b}^O(u_1,y_1) = 0$, then $\Phi^I(\beta^O)(u_1,y_1)$ is tangent to the homoclinic orbit at $(u_1,y_1)$. But the flow $\Phi^O$
maps tangent vectors to tangent vectors, so $\Phi^O(\Phi^I(\beta^O)(u_1,y_1))(u_{\textrm{max}},0)$ is vertical in this case, so in this case one can see directly that
$A_3 \geq 1$.

Lastly if $\Phi^I(\beta^O)(u_1,y_1) \times \textbf{b}^O(u_1,y_1) < 0$, then denote $\textbf{b}^O(u_{\textrm{max}},0) := (0,-a)$ where $a$ is some positive number. We then have
\begin{equation*}
0 > \Phi^O(\Phi^I(\beta^O)(u_1,y_1))(u_{\textrm{max}},0) \times (0,-a) = (-a) \Phi^O_1(\Phi^I(\beta^O)(u_1,y_1))(u_{\textrm{max}},0).
\end{equation*}
This means that the first coordinate of the image of the vector under the flow is positive, so the vector must have passed through the vertical at some point. Thus $A_1+A_2+A_3 \geq 2$ in this case as well.

So we have that for solutions of this type the number of zeros of the associated variational equation is greater than or equal to two. This completes the proof of the theorem.
\end{proof}

\begin{figure}[tbhp!]
\begin{center}
\includegraphics[scale=0.5]{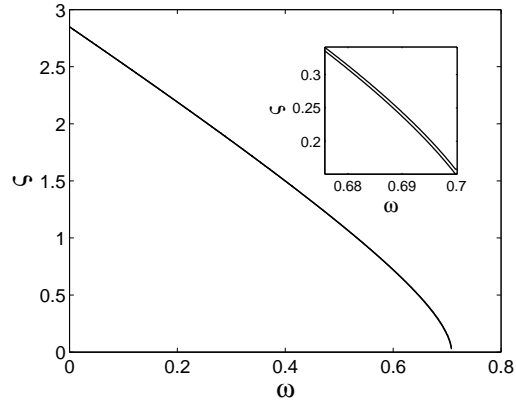}
\caption{The spectrum $\varsigma>0$ of $D_+$ for unstable symmetric solutions along the upper branch in Fig.\ \ref{fig2_num} as a function of $\omega$. The inset zooms in on a small region clearly showing that there are two curves.}
\label{fig2_P}
\end{center}
\end{figure}

\begin{remark}
\label{rem1}
Solutions shown in Fig.\ \ref{fig1_num}(b) and (c) (see also the corresponding phase-portraits in Fig.\ \ref{fig1_num_pp}(b) and (c)) belong respectively to case 1 and 2 in the proof of Theorem \ref{th:main} above. Let $\varsigma\in\R$ be the real part of a $\lambda$ in the spectrum of $D_+$ such that $D_+v(x)=\lambda v(x)$. In Fig.\ \ref{fig2_P} we plot the numerically obtained spectral parameter $\varsigma>0$ of symmetric solutions along the upper branch of Fig.\ \ref{fig2_num} as a function of $\omega$, from which one can see that $P=2$ indeed.
\end{remark}

So far we have been primarily concerned with positive solutions to (\ref{stat1}), but Theorem \ref{th:ckrtj88} applies equally well to standing waves $u(x)$ which are not strictly positive, so long as they are smooth enough and both $u$ and $u_x $ tend to 0 as $x \to \pm \infty$. Likewise, the techniques used to calculate $P$ in the proof of Theorem \ref{th:main} do not require that the solution be positive. As such, we can apply Theorem \ref{th:ckrtj88}, and the techniques above to establish the linear instability of some excited states with
$Q \geq 1$. The first excited standing wave we will deal with has $Q =1$.

\subsection{First excited unstable states}

\begin{figure}[tbp!]
\centering
\subfigure[]{\includegraphics[height= 2.0in]{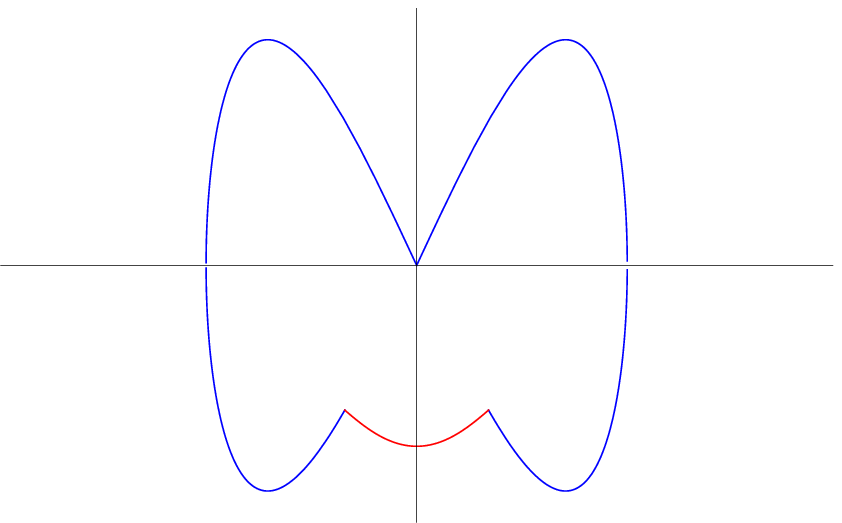}
\label{fig2:subfig1}
}
\subfigure[]{
\includegraphics[height= 2.0in]{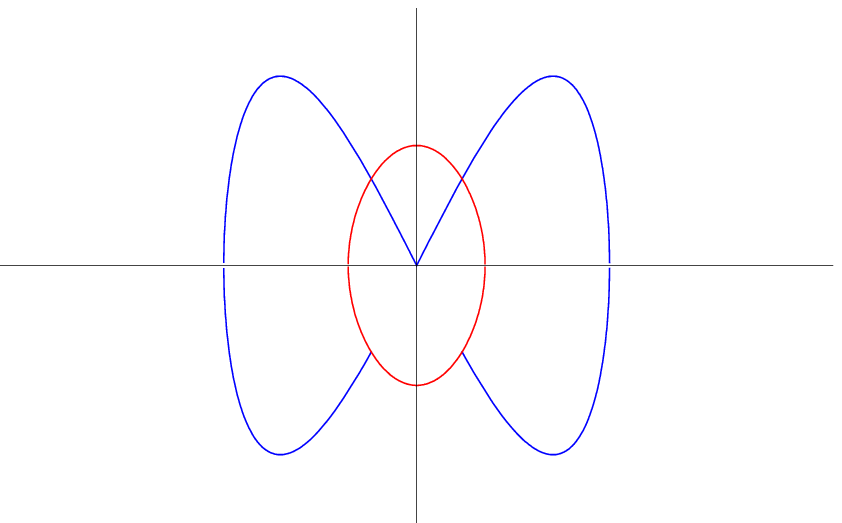}
\label{fig2:subfig2}
}
\label{fig2}
\caption{(a) An unstable orbit where $P=3$ and $Q =1$. (b) An unstable orbit where $P=2k+3$ and $Q =2k+1$ and $k \geq 1$. The number $k$ describes the number of periods the solution stays in the inner periodic orbit.}
\end{figure}

We will consider steady state solutions with the following properties
\begin{enumerate}
\item  The switch from the outer to the inner system takes place in the phase plane at $(u_0,y_0)$, with $y_0<0$ and $0<u_0 < \sqrt{\frac{V}{2}}$.
\item The solution returns to the outer homoclinic orbit at the point $(u_1,y_1) = (-u_0,y_0)$.
\end{enumerate}
Such a solution will be symmetric with respect to the $y$ axis in the phase plane, and such solutions exist for all values of $V$
and $\w$. An example of the phase portrait of such a solution can bee seen in figure \ref{fig2:subfig1}.

\begin{corollary}\label{cor:excited} Let $u(x)$ be as above, then $ P \geq 3$ and so we have a linearly unstable standing wave to equation (\ref{gov1}).
\end{corollary}

In order to prove the corollary, we need the following lemma, which is the analogue to lemma \ref{lem:asym} in the half of the phase plane where $u<0$.
\begin{lemma}\label{lem:asym2} If $\textbf{d}(u,y)$ is a vector in the tangent space to the phase plane of the homoclinic orbit at the point $(u,y)$, with $u<0$ such that
$\textbf{d}(u,y) \times \textbf{b}^O(u,y) > 0 $, then
\begin{equation}
\lim_{x \to \infty} \Phi^O(\frac{\textbf{d}}{|\textbf{d}|})(F^O(u,y)) = k \left(\begin{array}{ccc} 1 \\ \sqrt{V-\w } \end{array}\right)
\end{equation}
\noindent where $k$ is a negative real number.
\end{lemma}

The proof is exactly the same as in the proof of lemma \ref{lem:asym}, except that:
\begin{eqnarray}
\lim_{x \to \infty} \Phi^O(\frac{\textbf{b}^O)}{|\textbf{b}^O|})(F^O(u,y)) &= \lim_{x \to \infty} \frac{\textbf{b}^O}{|\textbf{b}^O|}(F^O(u,y))\nonumber\\ &= \frac{1}{-\sqrt{1+V-\w}}\left(\begin{array}{ccc} -1 \\ \sqrt{V-\w } \end{array}\right)
\end{eqnarray}
\noindent because we are in the left half plane. This change of sign therefore changes the sign of $k$.

It is straightforward to calculate that $A_1 \geq 1$. This is done in exactly the same way as in the proof of the second case of theorem \ref{th:main}.  Similarly, as in
the second case of theorem \ref{th:main}, one shows that $A_1+ A_2 + A_3 = P \geq 3$. This is done by using the cross product of the appropriate vectors under the
flows of the outer system, showing that the first coordinate of the tangent vector under the flow of the variational equation is first positive, then negative, then positive, and finally
ends up negative because of lemma \ref{lem:asym2}. These three sign changes mean that the vector has passed through the vertical at least three times, and hence
we have an unstable standing wave.

\begin{figure}[tbhp!]
\centering
\subfigure[]{\includegraphics[height= 2.5 in]{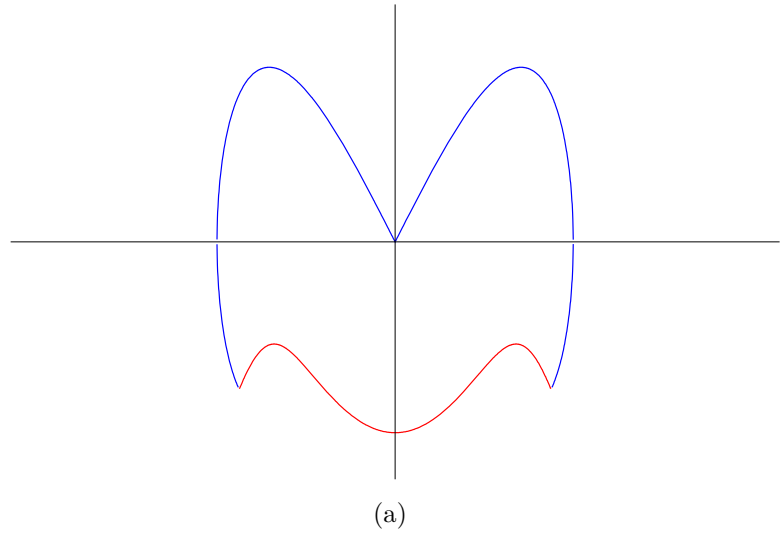}}
\subfigure[]{\includegraphics[scale= .5]{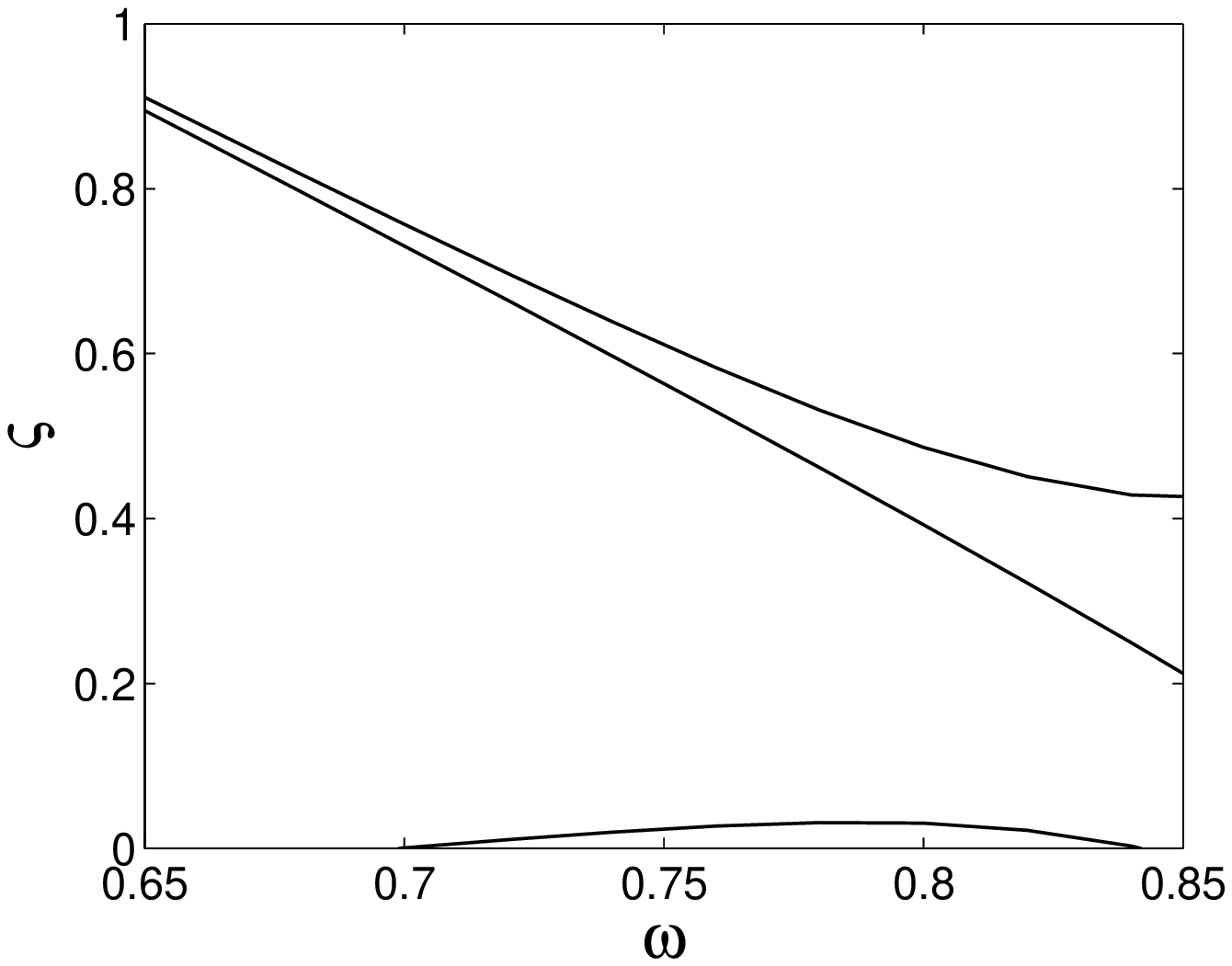}}
\caption{(a) A sketch of symmetric orbits with $P \geq 2$ and $Q =1$. (b) The same as in Fig.\ \ref{fig2_P}, but for unstable symmetric solutions in Fig.\ \ref{fig3_num}.}
\label{fig4:subfig1}
\end{figure}

\begin{remark}
\label{rem2}
The solution shown in Fig.\ \ref{fig4_num}(a) (see the corresponding phase-portrait in Fig.\ \ref{fig4_num_pp}(a)) belongs to the class of solutions discussed in Corollary \ref{cor:excited}. Nevertheless, the symmetric solution corresponding to point C in Fig.\ \ref{fig3_num} (cf.\ Fig.\ \ref{fig4_num}(c) and its phase-portrait \ref{fig4_num_pp}(c)) does not belong to the same class of solutions. 
In the phase-space, the orbit of the latter solution is sketched in Fig.\ \ref{fig4:subfig1}(a). This is similar to the configuration of the orbit in Corollary \ref{cor:excited}, which is a symmetric (about the $u$-axis) solution in the phase space with the properties $\frac{\w}{V} <\frac{2}{3} $, $y_0<0$ and $(u_1,y_1) = (-u_0, y_0)$, but here $u_0< \sqrt{\frac{V}{2}}$. In this case, the inner orbit lies outside the heteroclinic orbit of the inner system. In this configuration, we have that $Q = 1$ and $P \geq 2$. This implies that the above method cannot be used to prove the instability of such a solution. In Fig.\ \ref{fig4:subfig1}(b), we plot the positive spectrum $\varsigma>0$ of the operator $D_+$ for unstable symmetric solutions in Fig.\ \ref{fig3_num}, where one can see that $P=3$ only in some interval. We conjecture that the break down of our method occurs at the points where the branch containing point B and that containing points D and E emerge with the main branch corresponding to symmetric solutions.
\end{remark}

\subsection{Higher order excited states}

We remark that there are certain choices of $V$ and $\w$, and $u_0$ such that the inner part of the superimposed phase portrait will be part of a periodic orbit of the inner system. This applies for example if we keep the $(u_0,y_0)$ and $(u_1,y_1) = (-u_0, y_0)$ as in corollary \ref{cor:excited}, we can allow the trajectory to follow the periodic orbit for as many periods as we like, and can calculate the net effect on the stability of the standing wave as a whole. One can show that if the base point travels along the inner orbit described by (\ref{eq:inneru0}) from $(u_1,y_1)$  for exactly one period, the solution to the variational equation will have exactly 2 more zeros, as will the original solution. Thus for each period $P = Q = 2$, and so we still have an unstable excited solution. In this way we are capable of constructing unstable excited solutions with $2k+1$ zeros for an arbitrary $k$. Figure \ref{fig2:subfig2} is a phase portrait for such a standing wave solution.

It remains to determine which choices of $V$, $w$, and $u_0$ that will produce such an excited state. Using the geometry of the superposed phase portraits we have that such an excited steady state solution will exist when

\begin{equation}
\begin{array}{lll}
\frac{\w}{V} < \sqrt{\frac{1}{2}} & \textrm{and} & u_0 < \sqrt{\frac{{V - \sqrt{V^2-2\w^2}}}{2}} \textrm{,  or}\\
\frac{\w}{V} > \sqrt{\frac{1}{2}} & \textrm{and}  & u_0 < \sqrt{\frac{V}{2}}.
\end{array}
\end{equation}
\noindent The above technique can be used to create a higher order excited state whenever the inner part of the phase curve is part of a periodic orbit. We can simply follow the period as many times as is necessary and the higher order state will have the same difference of $P$ and $Q$. Thus if we begin with an unstable state where $P = P_0$ and $Q = Q_0$ being the initial values for $P$  and $Q$ and the inner part of the orbit is periodic (for another example c.f. figure(\ref{fig1:subfig1})), we can produce higher order unstable states with $ P = 2k+P_0$ and $Q = 2k+Q_0$, and for all integer values of $k$ we will still have an unstable state.

\section{Asymmetric states}

\begin{figure}[tbhp!]
\centering
\subfigure[]{
\includegraphics[height= 2.0in]{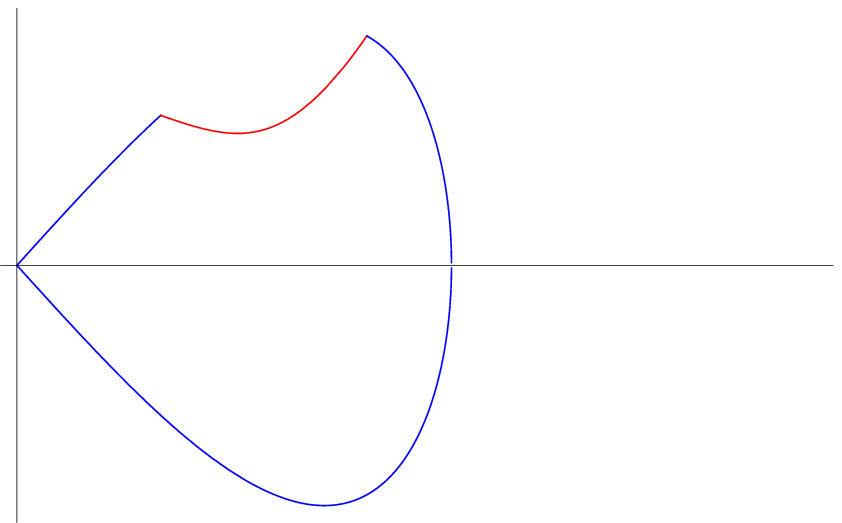}
\label{fig3:subfig1}
}
\subfigure[]{
\includegraphics[height= 2.0in]{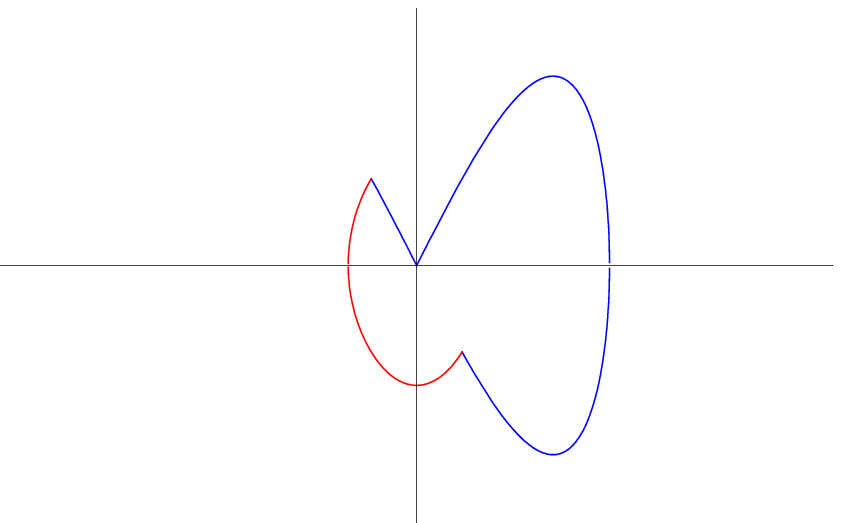}
\label{fig3:subfig2}
}
\label{fig3}
\caption{(a) An asymmetric orbit with $P \geq 1$ and $Q =0$. (b) An asymmetric orbit with $P = 2$ and $Q = 1$.}
\end{figure}

The techniques described and used in the section above cannot be applied to the asymmetric states considered in this paper because such states do not satisfy the condition in theorem \ref{th:ckrtj88}. As particular examples, we consider two solutions with $Q=0$ and $Q=1$.

For positive solutions $Q=0$, when $\frac{\w}{V} < \frac{3}{4}$,  asymmetric orbits will be present. An example of the phase portrait of a positive asymmetric solution is in Fig.\ \ref{fig3:subfig1}. For this case, it can be calculated that $P\geq1$. The profile shown in Fig.\ \ref{fig1_num}(c) with its phase-portrait depicted in Fig.\ \ref{fig1_num_pp}(c) belongs to this class of asymmetric solutions. As we know numerically that asymmetric positive solutions are stable (see Fig.\ \ref{fig1_num}), it is expected that Theorem \ref{th:ckrtj88} should not apply here. It is then interesting to know whether the stability can be established analytically.

Numerically (not shown here) we observed that $P=1$ for the asymmetric positive state. If one could show analytically that $P=1$ indeed, one would be able to prove the stability using the theorem presented in \cite{tran92b}. Nonetheless, the solution stability is expected and predicted using the method of \cite{mitc93}.

Excited asymmetric orbits where $P= 2$ and $Q = 1$ can occur for all configurations of $V$ and $w$. Such an orbit is described by the following properties. If
$\frac{\w}{V} < \frac{3}{2}$, then for the point of departure $(u_0, y_0)$, we have that $u_0 < \sqrt{\frac{V}{2}}$ and $y_0 < 0$. The orbit then travels through the $y$ axis in the phase plane and reconnects with the outer orbit at the point $(u_1,y_1)$ where $u_1 < 0$ and $y_1 > 0$.  Moreover, in this configuration of $V$
and $\w$, the only type of inner orbit that can produce such an asymmetric solution happens to be periodic and symmetric about the line $y = u$ in the phase plane. Thus $(u_1,y_1) = (-u_0,-y_0)$, and the original vector tangent to the homoclinic orbit ($\beta^O$) gets mapped by the variational flow of the inner orbit to the tangent space of the homoclinic orbit at the point $(u_1, y_1)$. Thus we know
\begin{equation*}
\lim_{x \to \infty} \frac{\Phi^O(\Phi^I(\beta^O)(u_1,y_1))(u,y)}{|\Phi^O(\Phi^I(\beta^O)(u_1,y_1))(u,y)|} = \left(\begin{array}{ccc} -k \\ k \sqrt{V-\w } \end{array}\right).
\end{equation*}
\noindent That is the limit under the variational flow of the vector which was originally tangent to the homoclinic orbit is the stable subspace of the variational flow at the
critical point at the origin. We can therefore only conclude that $P$ is exactly $2$.  If we are  in the case where
$\frac{\w}{V} > \frac{3}{2}$, and $0<u_0 < \sqrt{\frac{V}{2}}$, then we can exploit the same symmetry conditions as above for such asymmetric solutions, even
concluding the limit of the variational solution under the tangent flow as the stable subspace of the critical point at the origin. An example of one such orbit is in figure \ref{fig3:subfig2}, which is a sketch of the profile shown in Figs.\  \ref{fig4_num}(d) and (e). If however, $\frac{\w}{V} > \frac{3}{2}$ and
$u_0 > \sqrt{\frac{V}{2}}$, then no such periodic orbits exist, i.e. there is no connected inner orbit where $(u_0, y_0)$ and $(u_1,y_1)$ could have the aforementioned properties.

\section{Summary}

We have considered a nonlinear Schr\"odinger equation with a non-uniform nonlinearity coefficient. In particular, we have investigated a Schr\"odinger equation with self-defocusing nonlinearity bounded by self-focusing one. The present work extended the results of \cite{tran92b,mitc93,{jleon04}}. We have established analytically the instability of symmetric states beyond a critical norm through the application of a topological argument developed in \cite{ckrtj88}. Even though the technique does not definitively establish instability in the case of the asymmetric states considered in this paper, we have numerically established the stability of some positive asymmetric states. The analytical (in)stability of asymmetric higher-order modes remains an open problem, which is proposed to be studied in the future.

\appendix
\section{Phase portraits}

In Fig.\ \ref{fig1_num}, the solution profiles corresponding to points A--D in Fig.\ \ref{fig2_num} are shown in the physical space. The corresponding phase-portraits of the solutions are presented in Fig.\ \ref{fig1_num_pp}. The phase-portraits of solutions shown in Fig.\ \ref{fig4_num} are depicted in Fig.\ \ref{fig4_num_pp}.

\begin{figure}[tbhp!]
\begin{center}
\subfigure[]{\includegraphics[scale=0.45]{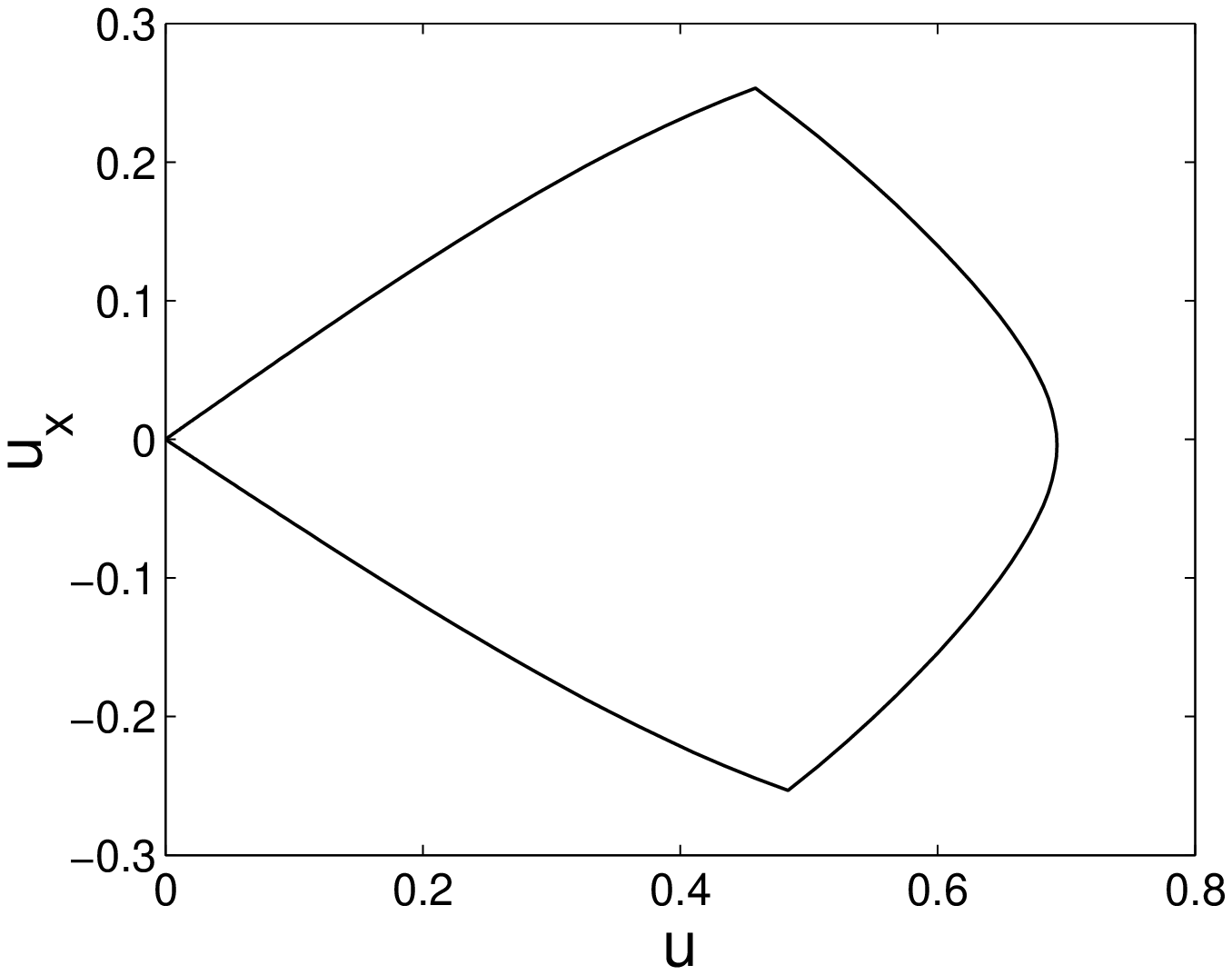}}
\subfigure[]{\includegraphics[scale=0.45]{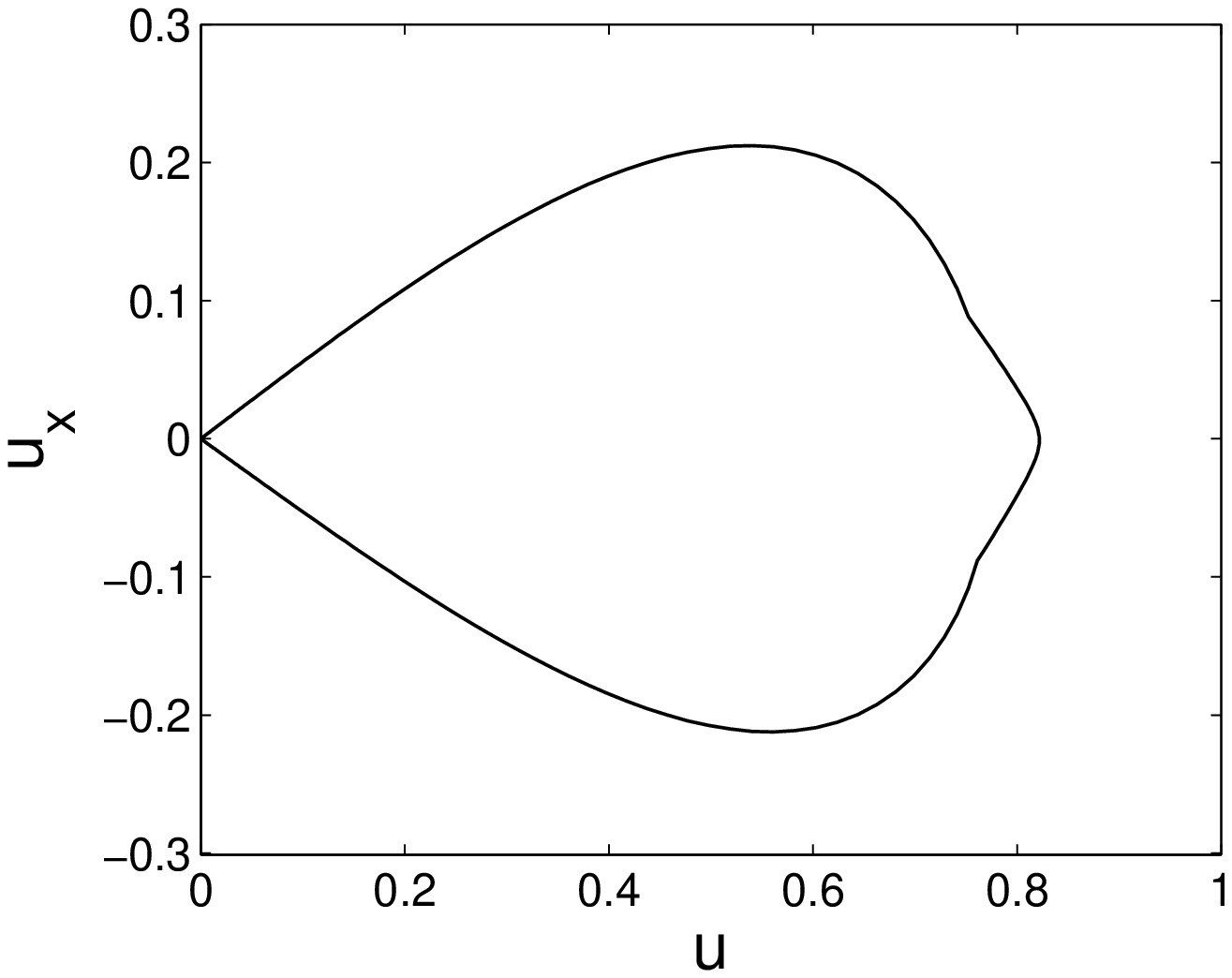}}
\subfigure[]{\includegraphics[scale=0.45]{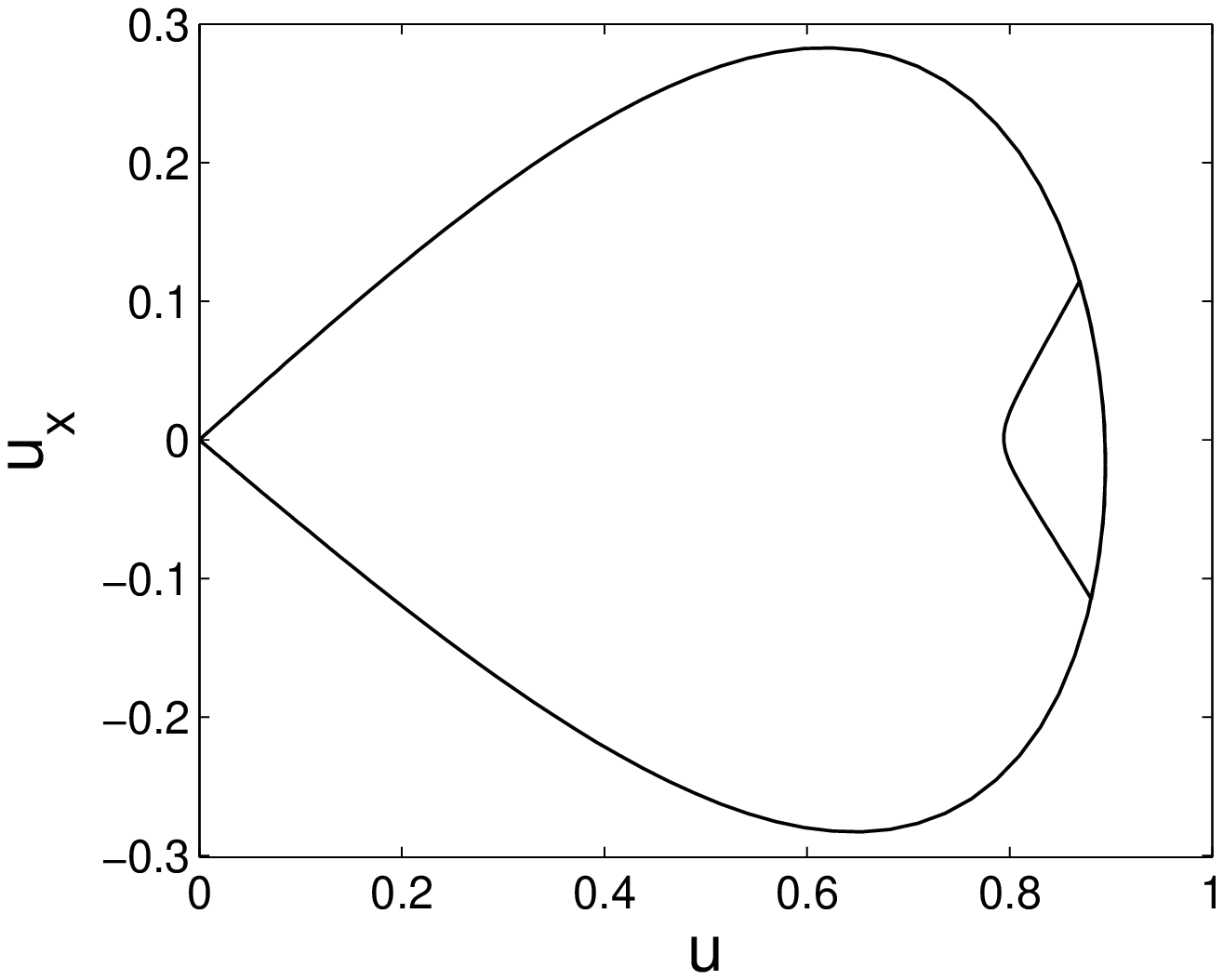}}
\subfigure[]{\includegraphics[scale=0.45]{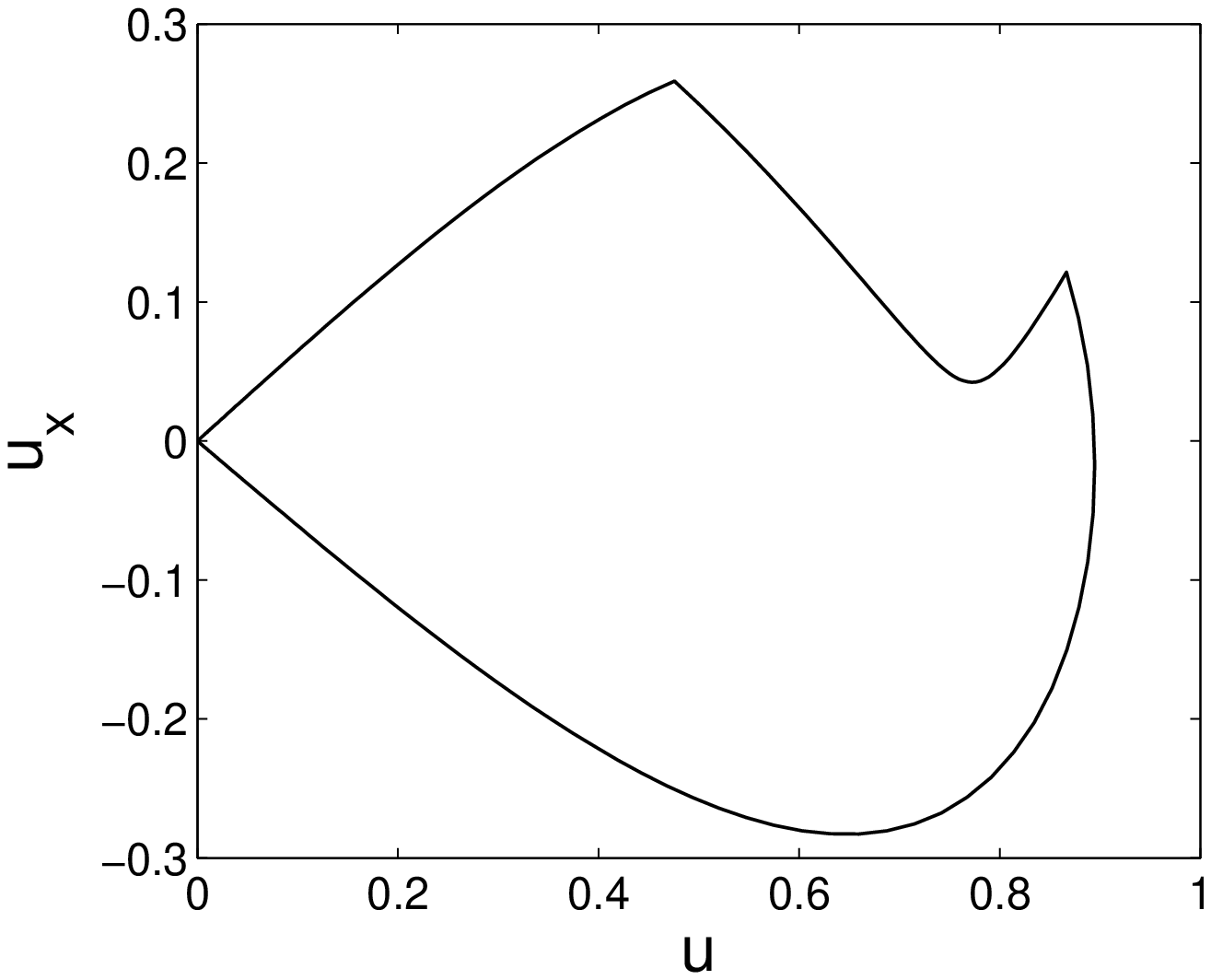}}
\caption{Phase portraits of solutions at points indicated as A--D in Fig.\ \ref{fig2_num} (see also Fig.\ \ref{fig1_num}).}
\label{fig1_num_pp}
\end{center}
\end{figure}

\begin{figure}[tbhp!]
\begin{center}
\subfigure[]{\includegraphics[scale=0.45]{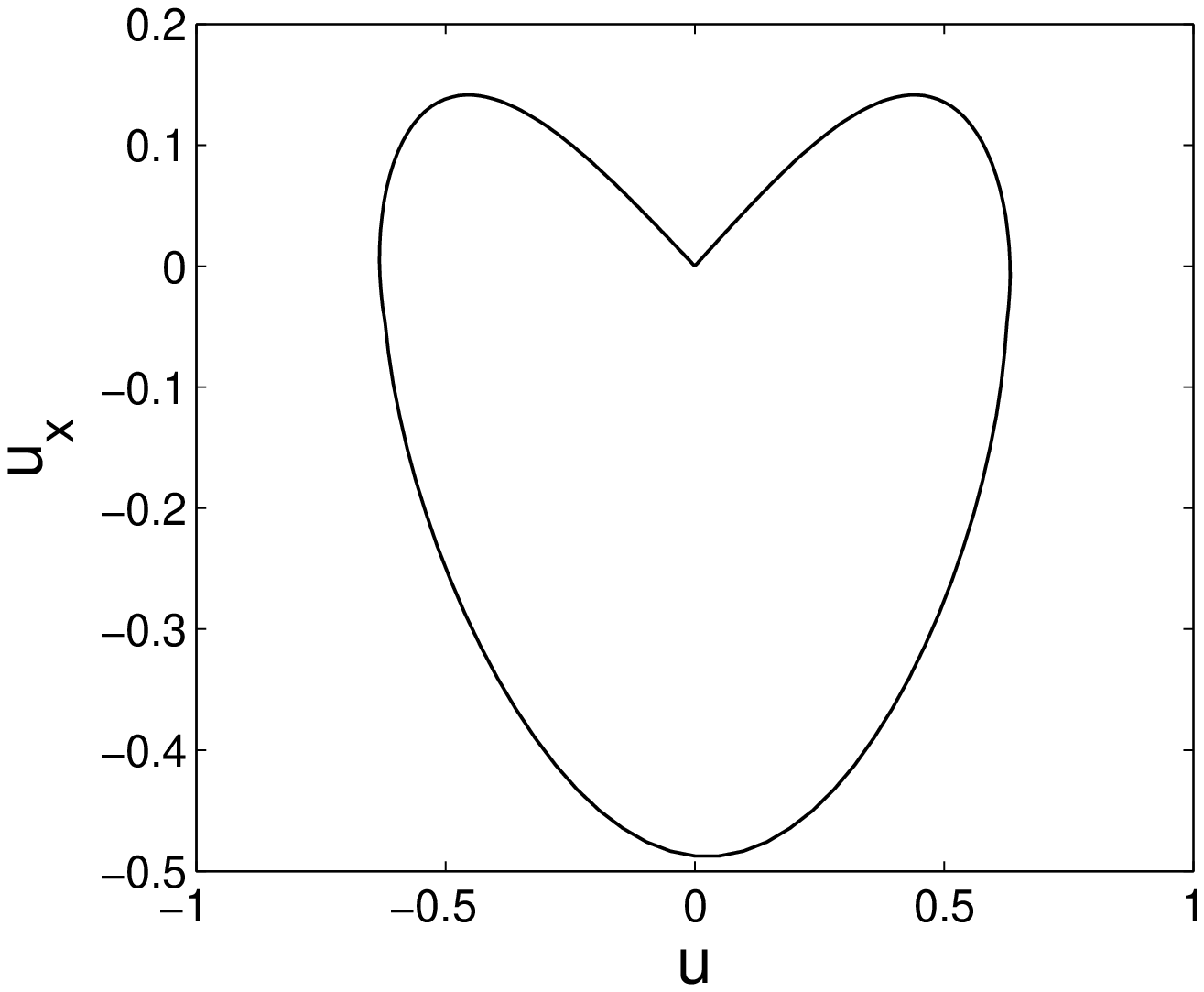}}
\subfigure[]{\includegraphics[scale=0.45]{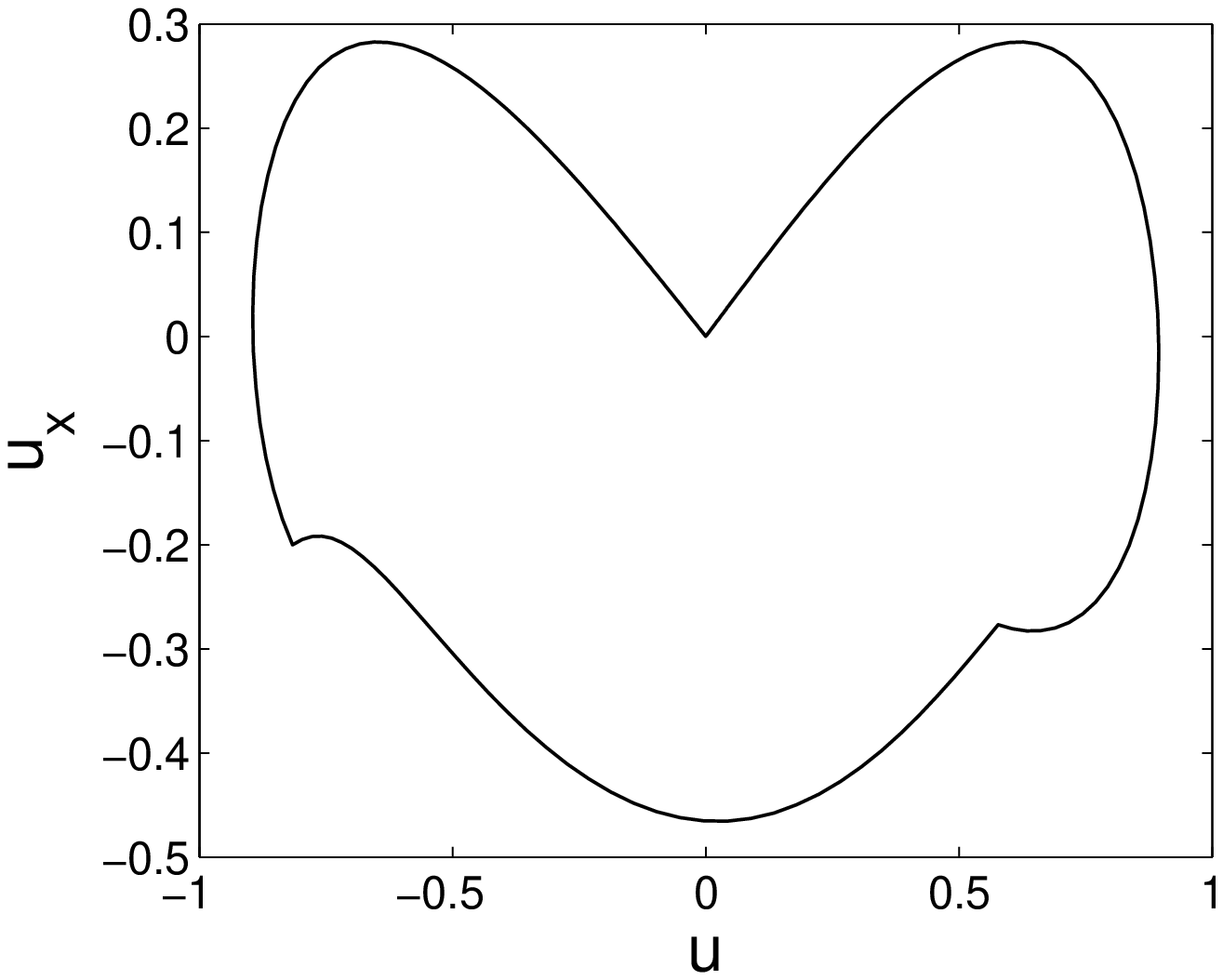}}
\subfigure[]{\includegraphics[scale=0.45]{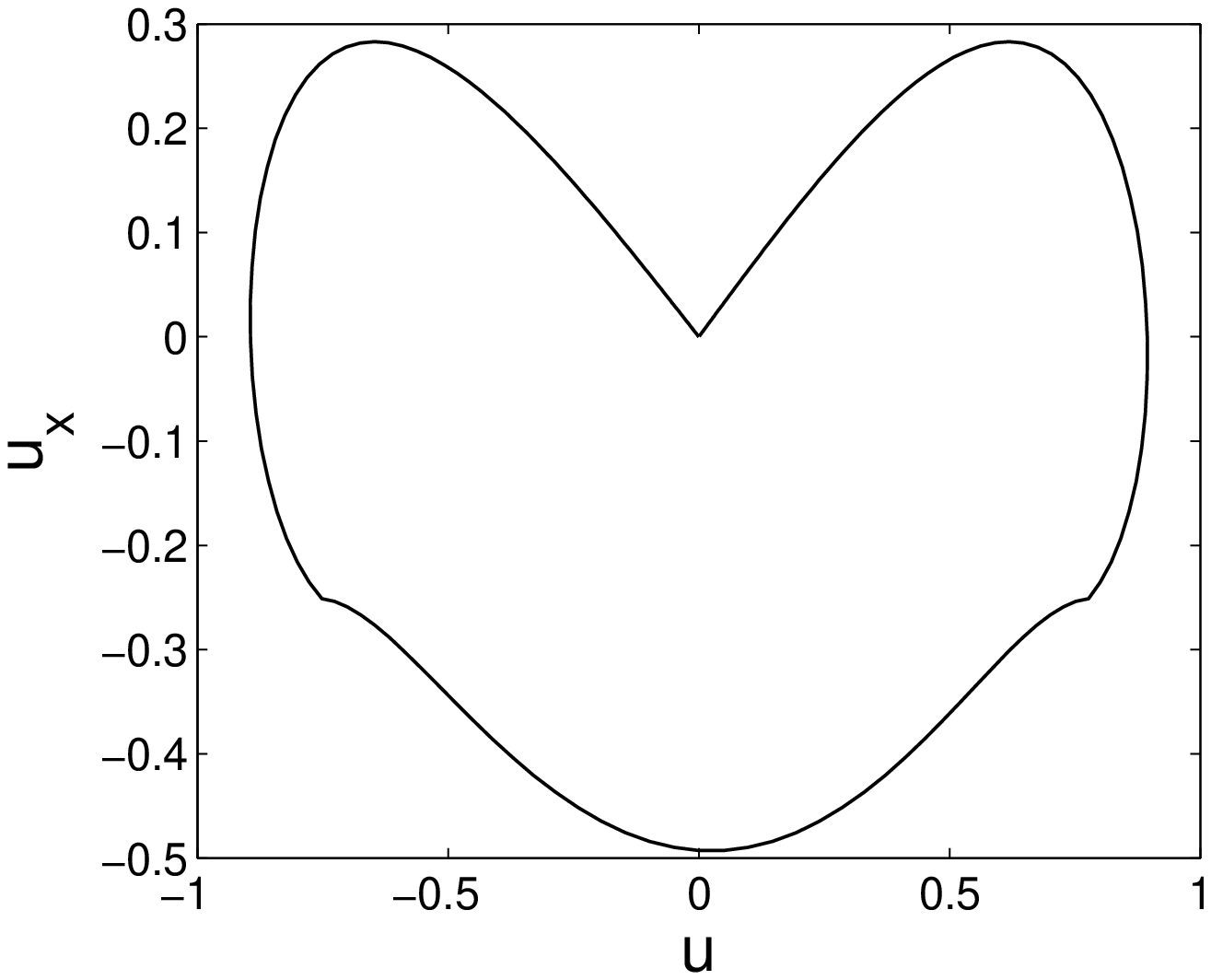}}
\subfigure[]{\includegraphics[scale=0.45]{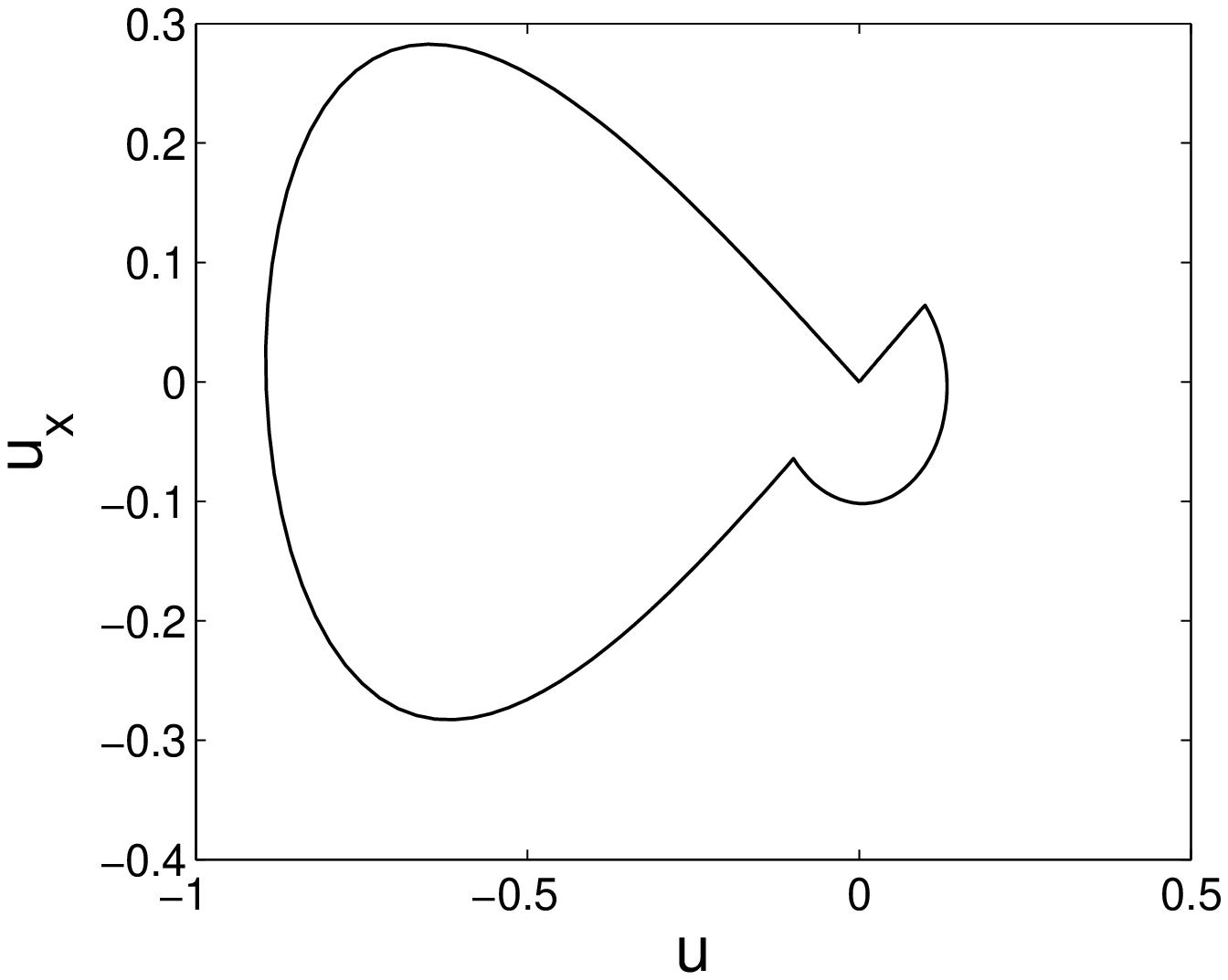}}
\subfigure[]{\includegraphics[scale=0.45]{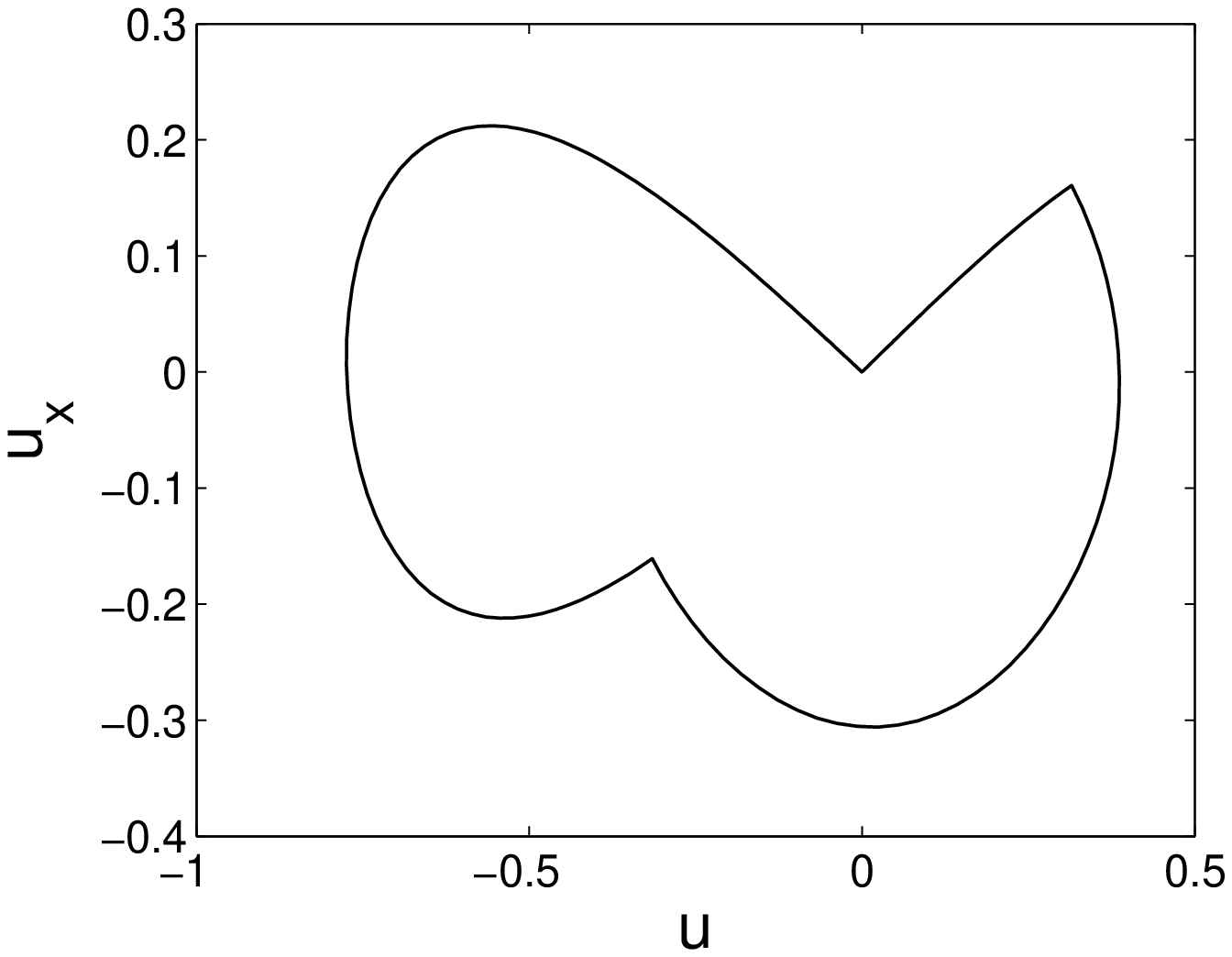}}
\caption{The same as Fig.\ \ref{fig1_num_pp} for solutions at points indicated as A--D in Fig.\ \ref{fig3_num} (see also Fig.\ \ref{fig4_num}). }
\label{fig4_num_pp}
\end{center}
\end{figure}

\end{document}